\documentclass[a4paper,UKenglish,cleveref,autoref]{lipics-v2019}

\usepackage{amssymb}
\usepackage{amsmath}
\usepackage{latexsym}
\usepackage{multicol}
\usepackage{booktabs}
\usepackage{tabularx}
\usepackage{multirow}
\usepackage{comment}
\usepackage{amsthm}

\newcommand{\sigmao}{\Sigma_{\mathsf{ortho}}}

\title{Parallel Computation of Alpha Complexes for Biomolecules} 

\author
{Talha Bin Masood\footnote{Corresponding author}}
{Scientific Visualization Group, Link\"oping University, Norrk\"oping, Sweden 
}
{talha.bin.masood@liu.se}
{https://orcid.org/0000-0001-5352-1086}
{}

\author
{Tathagata Ray}
{BITS Pilani, Hyderabad Campus, Hyderabad, India}
{rayt@hyderabad.bits-pilani.ac.in}
{}
{}

\author
{Vijay Natarajan}
{Department of Computer Science and Automation, Indian Institute of Science, Bangalore, India}
{vijayn@iisc.ac.in}
{https://orcid.org/0000-0002-7956-1470}
{This work is partially supported by a Swarnajayanti Fellowship from the Department of Science and Technology, India (DST/SJF/ETA-02/2015-16); a Mindtree Chair research grant; and the Robert Bosch Centre for Cyber Physical Systems, IISc.}

\authorrunning{T.\,B. Masood, T. Ray, and V. Natarajan}

\Copyright{Talha Bin Masood, Tathagata Ray and Vijay Natarajan}



\supplement{Source code available at: \url{https://bitbucket.org/vgl_iisc/parallelac/}}

\funding{}

\acknowledgements{Part of this work was done when the first author was at Indian Institute of Science, Bangalore. The authors would like to thank Sathish Vadhiyar and Nikhil Ranjanikar for helpful discussions and suggestions during the early phase of this work.}

\nolinenumbers 

\hideLIPIcs  

\EventEditors{Sergio Cabello and Danny Z. Chen}
\EventNoEds{2}
\EventLongTitle{36th International Symposium on Computational Geometry (SoCG 2020)}
\EventShortTitle{SoCG 2020}
\EventAcronym{SoCG}
\EventYear{2020}
\EventDate{June 23--26, 2020}
\EventLocation{Z\"{u}rich, Switzerland}
\EventLogo{socg-logo}
\SeriesVolume{164}
\ArticleNo{17}
\begin{document}

\maketitle

\begin{abstract}
The alpha complex, a subset of the Delaunay triangulation, has been extensively used as the underlying representation for biomolecular structures. We propose a GPU-based parallel algorithm for the computation of the alpha complex, which exploits the knowledge of typical spatial distribution and sizes of atoms in a biomolecule. Unlike existing methods, this algorithm does not require prior construction of the Delaunay triangulation. The algorithm computes the alpha complex in two stages. The first stage proceeds in a bottom-up fashion and computes a superset of the edges, triangles, and tetrahedra belonging to the alpha complex. The false positives from this estimation stage are removed in a subsequent pruning stage to obtain the correct alpha complex. Computational experiments on several biomolecules demonstrate the superior performance of the algorithm, up to a factor of 50 when compared to existing methods that are optimized for biomolecules.
\end{abstract}

\section{Introduction}
The alpha complex of a set of points in $\mathbb{R}^3$ is a subset of the Delaunay triangulation. A size parameter $\alpha$ determines the set of simplices (tetrahedra, triangles, edges, and vertices) of the Delaunay triangulation that are included in the alpha complex.  It is an elegant representation of the shape of the set of points~\cite{edelsbrunner1983shape, edelsbrunner1994three, alphaShapesSurvey2010} and has found various applications, particularly in molecular modeling and molecular graphics. The atoms in a biomolecule are represented by weighted points in $\mathbb{R}^3$, and the region occupied by the molecule is represented by the union of balls centered at these points. The geometric shape of a biomolecule determines its function, namely how it interacts with other biomolecules. The alpha complex represents the geometric shape of the molecule very efficiently. It has been widely used for computing and studying geometric features such as cavities and channels~\cite{liang1998analytical2, liang1998anatomy, dundas2006castp, Masood2015, Sridharamurthy2016, masood2016integrated, cavityReview2016}. Further, an alpha complex based representation is also crucial for accurate computation of geometric properties like volume and surface area~\cite{liang1998analytical1, edelsbrunner2005geometry, koehl2011volumes}. 

Advances in imaging technology have resulted in a significant increase in the size of molecular structure data. This necessitates the development of efficient methods for storing, processing, and querying these structures. In this paper, we study the problem of efficient construction of the alpha complex with particular focus on point distributions that are typical of biomolecules. In particular, we present a parallel algorithm for computing the alpha complex and an efficient GPU implementation that outperforms existing methods. In contrast to existing algorithms, our algorithm does not require the explicit construction of the Delaunay triangulation. 

\subsection{Related work}
The Delaunay triangulation has been studied within the field of computational geometry for several decades and numerous algorithms have been proposed for its construction~\cite{aurenhammer2013book}. Below, we describe only a few methods that are most relevant to this paper. 

A tetrahedron belongs to the Delaunay triangulation of a set of points in $\mathbb{R}^3$ if and only if it satisfies the \emph{empty circumsphere} property, namely no point is contained within the circumsphere of the tetrahedron.
The Bowyer-Watson algorithm~\cite{bowyer1981computing, watson1981computing} and the incremental insertion algorithm by Guibas~\textit{et al.}~\cite{guibas1992randomized} are based on the above characterization of the Delaunay triangulation. In both methods, points are inserted incrementally and the triangulation is locally updated to ensure that the Delaunay property is satisfied. The incremental insertion method followed by bi-stellar flipping works in higher dimensions also~\cite{edelsbrunner1996incremental} and can construct the Delaunay triangulation in $O(n\log n + n^{\left \lceil d/2 \right \rceil})$ time in the worst case, where $n$ is the number of input points in $\mathbb{R}^d$. A second approach for constructing the Delaunay triangulation is based on its equivalence to the convex hull of the points lifted onto a $(d+1)$-dimensional paraboloid~\cite{EdelsbrunnerHerbert1986}. 

A third divide-and-conquer approach partitions the  inputs points into two or generally multiple subsets, constructs the Delaunay triangulation for each partition, and merges the pieces of the triangulation finally. The merge procedure depends on the ability to order the edges incident on a vertex and hence works only in $\mathbb{R}^2$. The extension to $\mathbb{R}^3$ requires that the merge procedure be executed first~\cite{cignoni1998dewall}. The divide-and-conquer strategy directly extends to a parallel algorithm~\cite{Ashwin2012GPUDelaunay, cao2014gpu}. The DeWall algorithm~\cite{cignoni1998dewall} partitions the input point set into two halves and first constructs the triangulation of points lying within the boundary region of the two partitions. The Delaunay triangulation of the two halves is then constructed in parallel. The process is repeated recursively resulting in increased parallelism. Cao~\textit{et al.}~\cite{cao2014gpu} have developed a GPU parallel algorithm, \emph{gDel3D}, that constructs the Delaunay triangulation in two stages. In the first stage, points are inserted in parallel followed by flipping  to obtain an approximate Delaunay triangulation. In the second stage, a star splaying procedure works locally to convert non-Delaunay tetrahedra into Delaunay tetrahedra. The algorithm can be extended to construct the weighted Delaunay triangulation for points with weights. Cao~\textit{et al.} report a speed up of up to a factor of $10$ over a sequential implementation for constructing the weighted Delaunay triangulation of 3 million weighted points. 

Existing algorithms for constructing the alpha complex~\cite{edelsbrunner1992weighted, edelsbrunner1994three, koehl2011volumes, cgal:alpha3D} often require that the Delaunay triangulation be computed in a first step, with the exception of a recent method that guarantees output sensitive construction under mild assumptions on weights~\cite{Sheehy15} or a possible construction from \v{C}ech complexes~\cite{bauer2014morse}. Simplices that belong to the alpha complex are identified using a size filtration in a second step. Simplices that belong to the alpha complex are identified using a size filtration in a second step. In the case of biomolecules, only small values of the 
size parameter are of interest and the number of simplices in the alpha complex is a small fraction of those contained in the Delaunay triangulation. Hence, the Delaunay triangulation construction is often the bottleneck in the alpha complex computation. The key difficulty lies in the absence of a direct characterization of simplices that belong to the alpha complex. 

\subsection{Summary of results}
We propose an algorithm that avoids the expensive Delaunay triangulation computation and instead directly computes the alpha complex for biomolecules. The key contributions of this paper are summarized below:
\begin{itemize}
\item A new characterization of the alpha complex -- a set of conditions necessary and sufficient for a simplex to be a part of the alpha complex.
\item A new algorithm for computing the alpha complex of a set of weighted points in $\mathbb{R}^3$. The algorithm identifies simplices of the alpha complex in decreasing order of dimension without computing the complete weighted Delaunay triangulation.
\item An efficient CUDA-based parallel implementation of this algorithm for biomolecular data that can compute the alpha complex for a 10 million point dataset in approximately 10 seconds.
\item A proof of correctness of the algorithm and comprehensive experimental validation to demonstrate that it outperforms existing methods.
\end{itemize}
While the experimental results presented here focus on biomolecular data, the algorithm is applicable to data from other application domains as well. In particular, the efficient GPU implementation may be used for points that arise in smoothed particle hydrodynamics (SPH) simulations, atomistic simulations in material science, and particle systems that appear in computational fluid dynamics (CFD).

\section{Background}
In this section, we review the necessary background on Delaunay triangulations required to describe the algorithm and also establish a new characterization of the alpha complex that does not require the Delaunay triangulation. For a detailed description of Delaunay triangulations, alpha complexes, and related structures, we refer the reader to various books on the topic~\cite{aurenhammer2013book, edelsbrunnerMeshBook2001, edelsbrunner2010Computational}. 

\begin{figure}[ht!]	
\centering	
\subcaptionbox{\label{fig:disks}}{	
    \includegraphics[width=4.5cm]{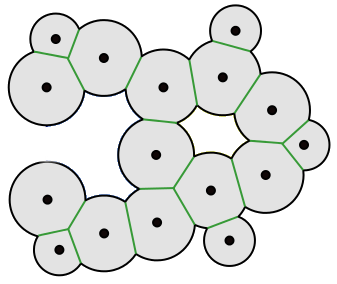}	
}	
\qquad	
\subcaptionbox{\label{fig:PD1}}{	
    \includegraphics[width=4.5cm]{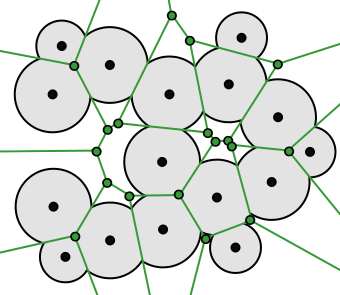}	
}	
\\	
\subcaptionbox{\label{fig:PD2}}{	
    \includegraphics[width=4.5cm]{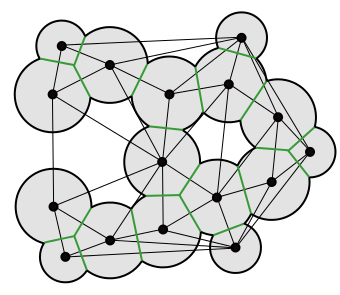}	
}	
\qquad	
\subcaptionbox{\label{fig:alpha0}}{	
    \includegraphics[width=4.5cm]{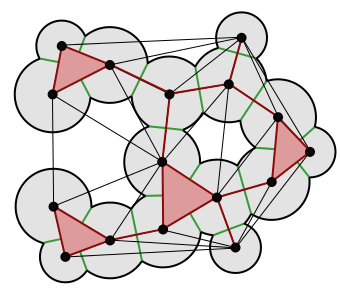}	
}	
\\	
\subcaptionbox{\label{fig:alphalater}}{	
    \includegraphics[width=4.5cm]{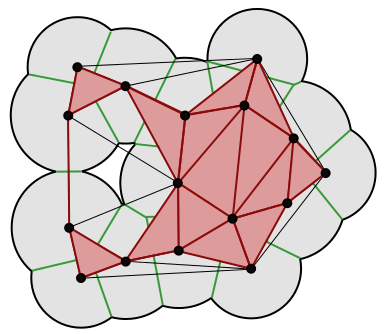}	
}	
\caption[2D illustration of alpha complex]{	
2D weighted Delaunay triangulation and alpha complex.	
\subref{fig:disks}~A set of weighted points $B$ in $\mathbb{R}^2$ shown as disks.	
\subref{fig:PD1}~The weighted Voronoi diagram of $B$. Voronoi edges and vertices are highlighted in green. 	
\subref{fig:PD2}~The weighted Delaunay complex is the dual of the weighted Voronoi diagram.	
\subref{fig:alpha0}~The alpha complex $K_\alpha$ for $\alpha = 0$ is shown in red. This is the dual of the intersection of the weighted Voronoi diagram and union of balls. 	
\subref{fig:alphalater}~The alpha complex shown for an $\alpha>0$. It is the dual of the intersection of the weighted Voronoi diagram and union of balls after growing them to have radius $\sqrt{r_i^2+\alpha}$.	
}	
\label{fig:background}	
\end{figure}

Let $B = \{b_i\}$ denote a set of balls or weighted points, where $b_i = (p_i , r_i)$ represents a ball centered at $p_i$ with radius $r_i$. We limit our discussion to balls in $\mathbb{R}^3$, so $p_i =(x_i , y_i , z_i) \in \mathbb{R}^3$. Further, we assume that the points in $B$ are in general position, \emph{i.e.}, no two points have the same location, no three points are collinear, no four points are coplanar, and no subset of five points are equidistant from a point in $\mathbb{R}^3$. Such configurations are called \emph{degeneracies}. In practice, a degenerate input can be handled via symbolic perturbation~\cite{edelsbrunner1990simulation}.

\subsection{Simplex and simplicial complex}
A $d$-dimensional \emph{simplex} $\sigma^d$ is defined as the convex hull of $d+1$ affinely independent points. Assuming the centres of balls in $B$ are in general position, all $(d + 1)$ sized subsets of $B$ form a simplex $\sigma^d=(p_0^\sigma,p_1^\sigma,\cdots,p_d^\sigma)$. For simplicity, we sometimes use $b_i$ instead of the center $p_i$ to refer to points incident on a simplex. For example, we may write $\sigma^d=(b_0^\sigma,b_1^\sigma,\cdots,b_d^\sigma)$. 

A non-empty strict subset of $\sigma^d$ is also a simplex but with dimension smaller than $d$. Such a simplex is called a \emph{face} of $\sigma^d$. Specifically, a $(d-1)$-dimensional face of $\sigma^d$ is referred to as a \emph{facet} of $\sigma^d$. A set of simplices $K$ is called a \emph{simplicial complex} if: 1)~a simplex $\sigma \in K$ implies that all faces of $\sigma$ also belong to $K$, and 2)~for two simplices $\sigma_1$, $\sigma_2 \in K$, either $\sigma_1\cap\sigma_2\in K$ or $\sigma_1\cap\sigma_2 = \emptyset$.

\subsection{Power distance and weighted Voronoi diagram}
The \emph{power distance} $\pi(p, b_i)$ between a point $p\in\mathbb{R}^3$ and a ball $b_i = (p_i, r_i) \in B$ is defined as
$$\pi(p, b_i) = \Vert p - p_i\Vert^2-r_i^2.$$
The \emph{weighted Voronoi diagram} is an extension of the Voronoi diagram to weighted points. It is a partition of $\mathbb{R}^3$ based on proximity to input balls $b_i$ in terms of the power distance. Points $p \in \mathbb{R}^3$ that are closer to the ball $b_i$ compared to all other balls $b_j \in B$ ($j \neq i$) constitute the \emph{Voronoi region} of $b_i$. Points equidistant from two balls $b_i$, $b_j \in B$ and closer to these two balls compared to other balls constitute a \emph{Voronoi face}. Similarly, points equidistant from three balls and fours balls constitute \emph{Voronoi edges} and \emph{Voronoi vertices} of the weighted Voronoi diagram, respectively. Figure~\ref{fig:PD1} shows the weighted Voronoi diagram for a set of 2D weighted points or disks on the plane. Similar to the unweighted case, the Voronoi regions of the weighted Voronoi diagram are convex and linear. However, the weights may lead to a configuration where the Voronoi region of $b_i$ is disjoint from $b_i$. This occurs when $b_i$ is contained within another ball $b_j$. Further, the Voronoi region of $b_i$ may even be empty.

\subsection{Weighted Delaunay triangulation}
The \emph{weighted Delaunay triangulation} is the dual of the weighted Voronoi diagram, see Figure~\ref{fig:PD2}. It is a simplicial complex consisting of simplices that are dual to the cells of the weighted Voronoi diagram. The following equivalent definition characterizes a simplex $\sigma^d$ belonging to a Delaunay triangulation $D$.

\begin{definition}[Weighted Delaunay Triangulation] A simplex $\sigma^d$ = $(p_0^\sigma,p_1^\sigma,\cdots,p_d^\sigma),\,  0 \leq d \leq 3$, belongs to the weighted Delaunay triangulation $D$ of $B$ if and only if there exists a point $p \in \mathbb{R}^3$ such that
\begin{description}
\item[DT1:] $\pi(p, b_0^\sigma)=\pi(p, b_1^\sigma)=\cdots=\pi(p, b_d^\sigma)$, and
\item[DT2:] $\pi(p, b_0^\sigma)\leq \pi(p, b_i)$ for $b_i \in B-\sigma^d$.
\end{description}
\end{definition}

\begin{figure}[hb!]	
\centering	
\subcaptionbox{\label{fig:selDiscs}}{\includegraphics[width=6cm]{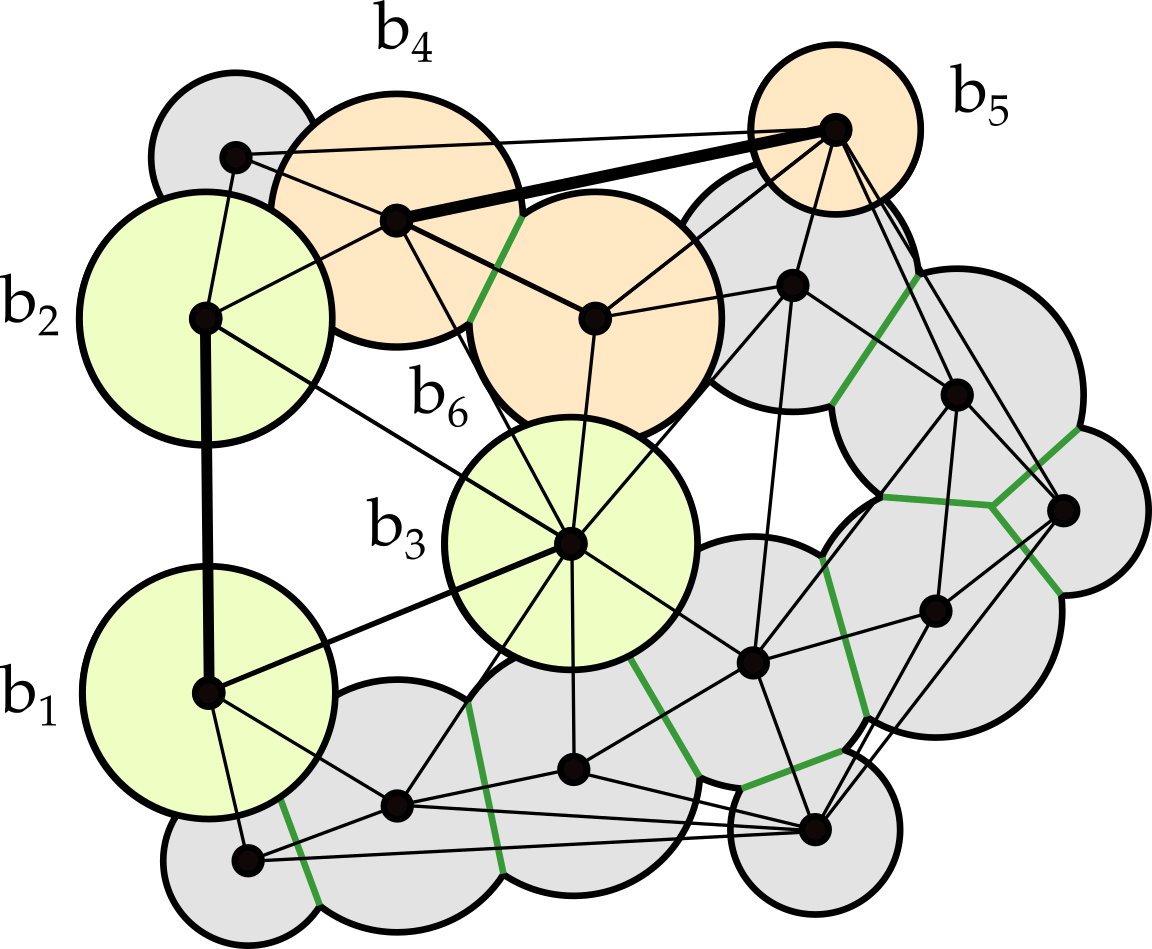}	
}	
\\	
\subcaptionbox{\label{fig:case1}}{\includegraphics[width=6cm]{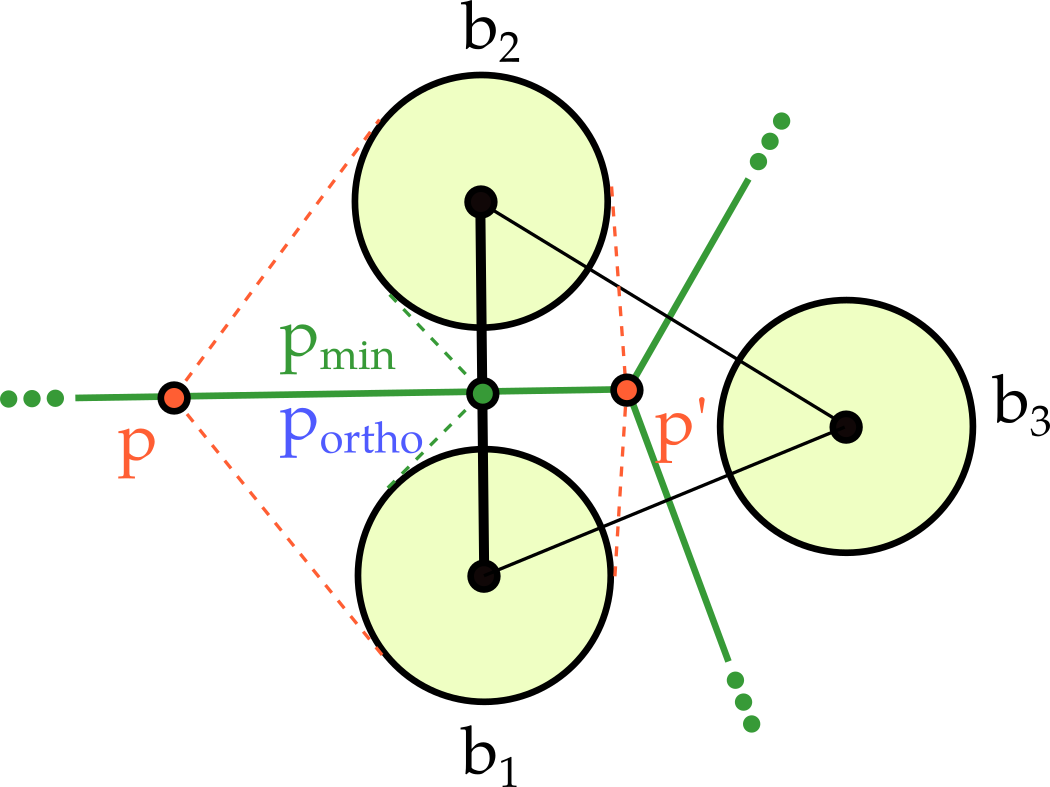}}	
\qquad	
\subcaptionbox{\label{fig:case2}}{\includegraphics[width=5cm]{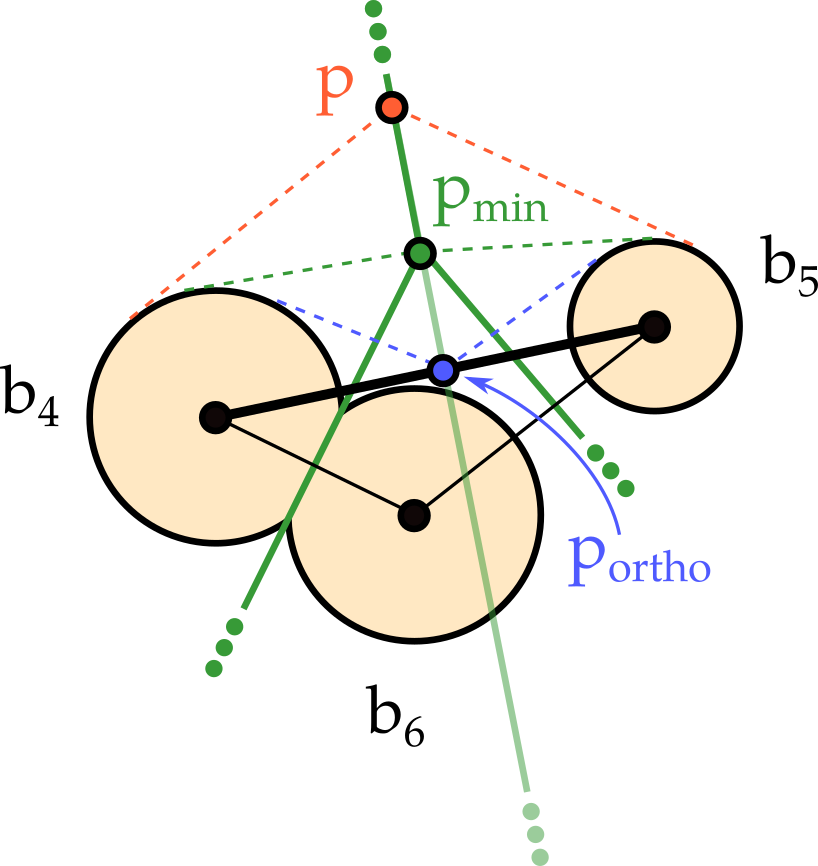}}	
\caption[The two cases]{$\mathsf{Size}$ and $\mathsf{OrthoSize}$ of a simplex. 
\subref{fig:selDiscs}~A set $B$ of weighted points. Two edges~(bold) belong to the Delaunay triangulation. 
\subref{fig:case1}~The $\mathsf{Size}$ of edge $b_1b_2$ is equal to its $\mathsf{OrthoSize}$. Points $p$, $p'$, $p_{\mathsf{min}}$ and $p_{\mathsf{ortho}}$ are \emph{witnesses}. Each one is equidistant from $b_1$ and $b_2$ and farther away from other disks in $B$. The distance is proportional to the length of the tangent to the disk that represents the weighted point. The next closest disk from these points is $b_3$. In this case, $p_{\mathsf{min}}$ and $p_{\mathsf{ortho}}$ coincide and hence $\mathsf{Size} = \mathsf{OrthoSize}$.
\subref{fig:case2}~$b_4b_5$ is also a Delaunay edge. The location of a neighboring disk $b_6$ could lead to a different configuration. The point $p_{\mathsf{ortho}}$ is closest to $b_4$ and $b_5$ among all the points that are equidistant from both. However $p_{\mathsf{ortho}}$ is closer to $b_6$ as compared to $b_4$ and $b_5$. The closest point $p_{\mathsf{min}}$ that satisfies DT1 and DT2 is farther away, hence for $b_4b_5$ $\mathsf{Size}$ 
is greater than $\mathsf{OrthoSize}$.	
}	
\label{fig:twoCases}	
\end{figure}

A point $p$ that satisfies the above two conditions, DT1 and DT2, is called a \emph{witness} for $\sigma^d$. We call a point that minimizes the distance $\pi(p, b_0^\sigma)$ and satisfies both conditions as the \emph{closest witness}, denoted by $p_{\mathsf{min}}^\sigma$. This minimum distance $\pi(p_{\mathsf{min}}^\sigma, b_0^\sigma)$ is called the $\mathsf{Size}$ of the simplex $\sigma^d$. A point that minimizes the distance $\pi(p, b_0^\sigma)$ and satisfies DT1 is called the $ortho\text{-}center$~$p_{\mathsf{ortho}}^\sigma$ of simplex $\sigma^d$. The distance $\pi(p_{\mathsf{ortho}}^\sigma, b_0^\sigma)$ is called the $\mathsf{OrthoSize}$ of the simplex $\sigma^d$. Clearly, the $\mathsf{Size}$ of a simplex is lower bounded by its $\mathsf{OrthoSize}$. Figure~\ref{fig:twoCases} shows the two possible scenarios, namely when $\mathsf{OrthoSize}=\mathsf{Size}$ and $\mathsf{OrthoSize}<\mathsf{Size}$.

\subsection{Alpha complex} Given a parameter $\alpha \in \mathbb{R}$, we can construct a subset of the weighted Delaunay triangulation by filtering simplices whose $\mathsf{Size}$ is less than or equal to $\alpha$, see Figures~\ref{fig:alpha0}~and~\ref{fig:alphalater}. The resulting subset, called the \emph{alpha complex}, is a subcomplex of the Delaunay complex and is denoted $K_\alpha$:  $$K_\alpha = \{\sigma^d \in D \quad \text{such~that} \quad \mathsf{Size}(\sigma^d) \leq \alpha\}.$$
The following equivalent definition characterizes simplices of the alpha complex without explicitly referring to the Delaunay triangulation. 
\begin{definition}[Alpha complex] A $d$-dimensional simplex $\sigma^d$ = $(p_0^\sigma,p_1^\sigma,\cdots,p_d^\sigma),  0 \leq d \leq 3$, belongs to the alpha complex $K_\alpha$  of $B$ if and only if there exists a point $p \in \mathbb{R}^3$ such that the following three conditions are satisfied:
\begin{description}
\item[AC1:] $\pi(p, b_0^\sigma)=\pi(p, b_1^\sigma)=\cdots=\pi(p, b_d^\sigma)$,
\item[AC2:] $\pi(p, b_0^\sigma)\leq \pi(p, b_i)$ for $b_i \in B-\sigma^d$, and
\item[AC3:] $\pi(p, b_0^\sigma)\leq \alpha$ or equivalently, the $\mathsf{Size}$ of $\sigma^d$ is at most $\alpha$.
\end{description}
\end{definition}

\section{Algorithm}
\label{sec:alphaAlgo}
We now describe an algorithm to compute the alpha complex and prove its correctness. The algorithm utilizes the characterizing conditions introduced above. It first identifies the tetrahedra that belong to the alpha complex, followed by the set of triangles, edges and vertices. Figure~\ref{fig:algoDemo} illustrates the algorithm as applied to disks on the plane.
 
\subsection{Outline}
The alpha complex of a point set in $R^3$ consists of simplices of dimensions 0--3, $K_\alpha = K_\alpha^0 \cup K_\alpha^1 \cup K_\alpha^2 \cup K_\alpha^3$, where $K_\alpha^d \subset K_\alpha$ is the set of $d$-dimensional simplices in $K_\alpha$. We initialize $K_\alpha^d = \emptyset$ and construct $K_\alpha$ in five steps described below:
\begin{description}
\item[Step 1:] For $0 \leq d \leq 3$, compute the set of all simplices $\sigma^d$ such that $\mathsf{Ortho\mathsf{Size}}(\sigma^d)\leq \alpha$. Let this set be denoted by $\sigmao = \sigmao^0 \cup \sigmao^1 \cup \sigmao^2 \cup \sigmao^3$.
\item[Step 2:] For all tetrahedra $\sigma^3 \in \sigmao^3$, check condition AC2 using $p = p_{\mathsf{ortho}}^\sigma$. If $\sigma^3$ satisfies AC2 then insert it into $K_\alpha^3$.  
\item[Step 3:] Insert all triangles that are incident on tetrahedra in $K_\alpha^3$ into $K_\alpha^2$. Let $\Sigma_{\mathsf{free}}^2=\sigmao^2-\mathsf{Facets}(K_\alpha^3)$, where $\mathsf{Facets}(K_\alpha^3)$ 
denotes the set of facets of tetrahedra in $K_\alpha^3$. For all triangles $\sigma^2 \in \Sigma_{\mathsf{free}}^2$, check condition AC2 using $p = p_{\mathsf{ortho}}^\sigma$. If $\sigma^2$ satisfies AC2 then insert it into $K_\alpha^2$. 
\item[Step 4:] Insert all edges incident on triangles in $K_\alpha^2$ into $K_\alpha^1$. Let $\Sigma_{\mathsf{free}}^1=\sigmao^1-\mathsf{Facets}(K_\alpha^2)$, where $\mathsf{Facets}(K_\alpha^2)$ denotes the set of facets of triangles in $K_\alpha^2$. For all edges $\sigma^1 \in \Sigma_{\mathsf{free}}^1$, check condition AC2 using $p = p_{\mathsf{ortho}}^\sigma$. If $\sigma^1$ satisfies AC2 then insert it into $K_\alpha^1$.
\item[Step 5:] Insert all endpoints of edges in $K_\alpha^1$ into $K_\alpha^0$. Let $\Sigma_{\mathsf{free}}^0=\sigmao^0-\mathsf{Facets}(K_\alpha^1)$, where $\mathsf{Facets}(K_\alpha^1)$ denotes the set of balls incident on edges in $K_\alpha^1$. For all balls $b_i =(p_i,r_i) \in \Sigma_{\mathsf{free}}^0$, check condition AC2 using $p = p_i$. If $p_i$ satisfies AC2 then insert it into $K_\alpha^0$.
\end{description}

Step~1 selects simplices that satisfy AC3. Step~2 recognizes tetrahedra that belong to the alpha complex by checking AC2 using $p = p_{\mathsf{ortho}}^\sigma$. Triangle faces of these tetrahedra also belong to $K_\alpha$. The other ``free'' triangles belong to $K_\alpha^2$ if they satisfy AC2. Step~4 identify edges similarly. First all edge faces of triangles in $K_\alpha^2$ are inserted followed by those ``free'' edges that satisfy AC2. Vertices are identified similarly in Step~5.

A notion related to free simplices, called \textit{unattached} simplices, was introduced by Edelsbrunner~\cite{edelsbrunner1992weighted}. However, the characterization of unattached simplices depends on the fact that they belong to the Delaunay complex.
\begin{figure}[ht!]
\centering
\subcaptionbox{\label{fig:inputAlpha}}{	
    \includegraphics[width=4cm]{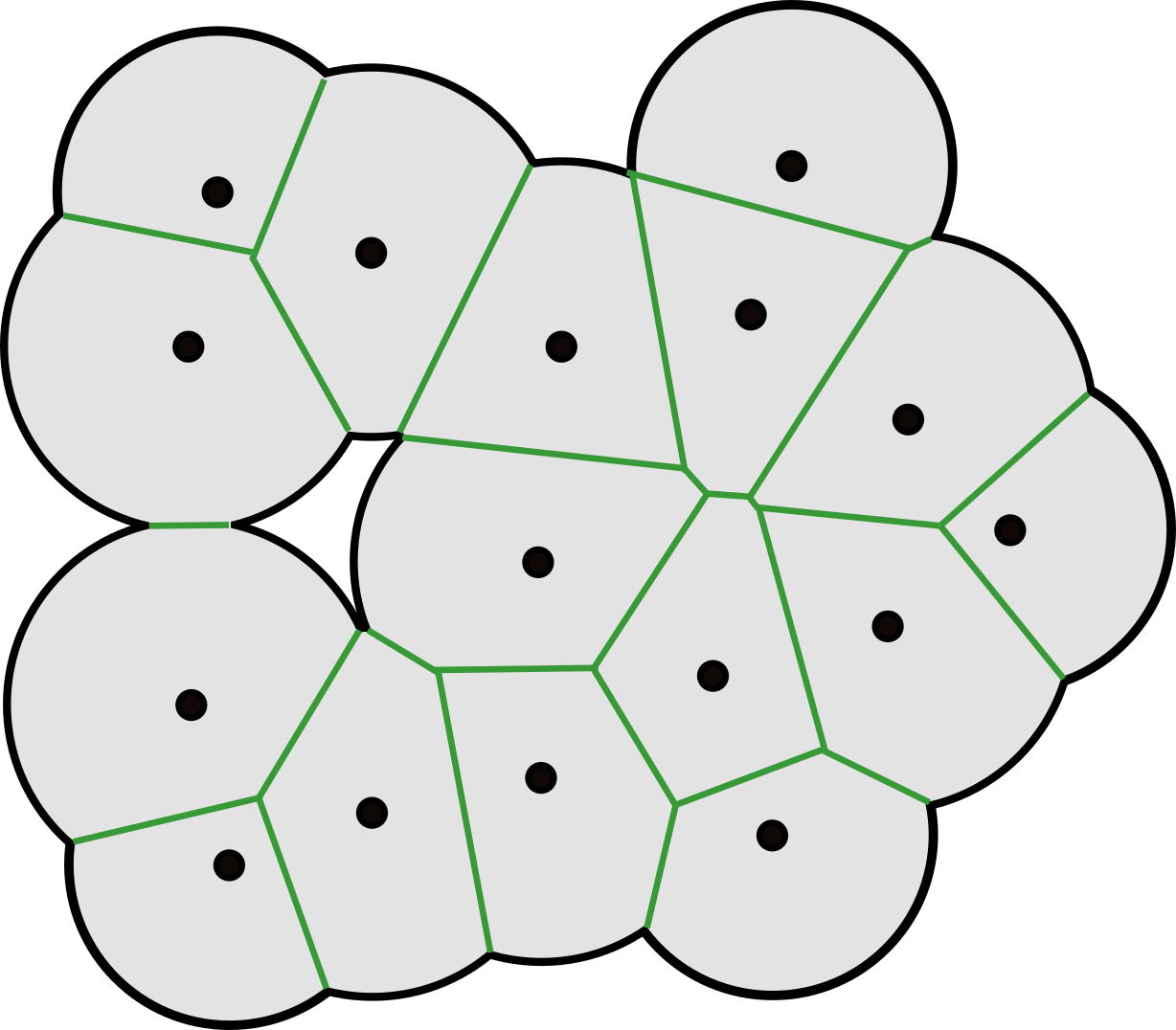}	
}	
\qquad	
\subcaptionbox{\label{fig:initAlphaEdges}}{	
    \includegraphics[width=4cm]{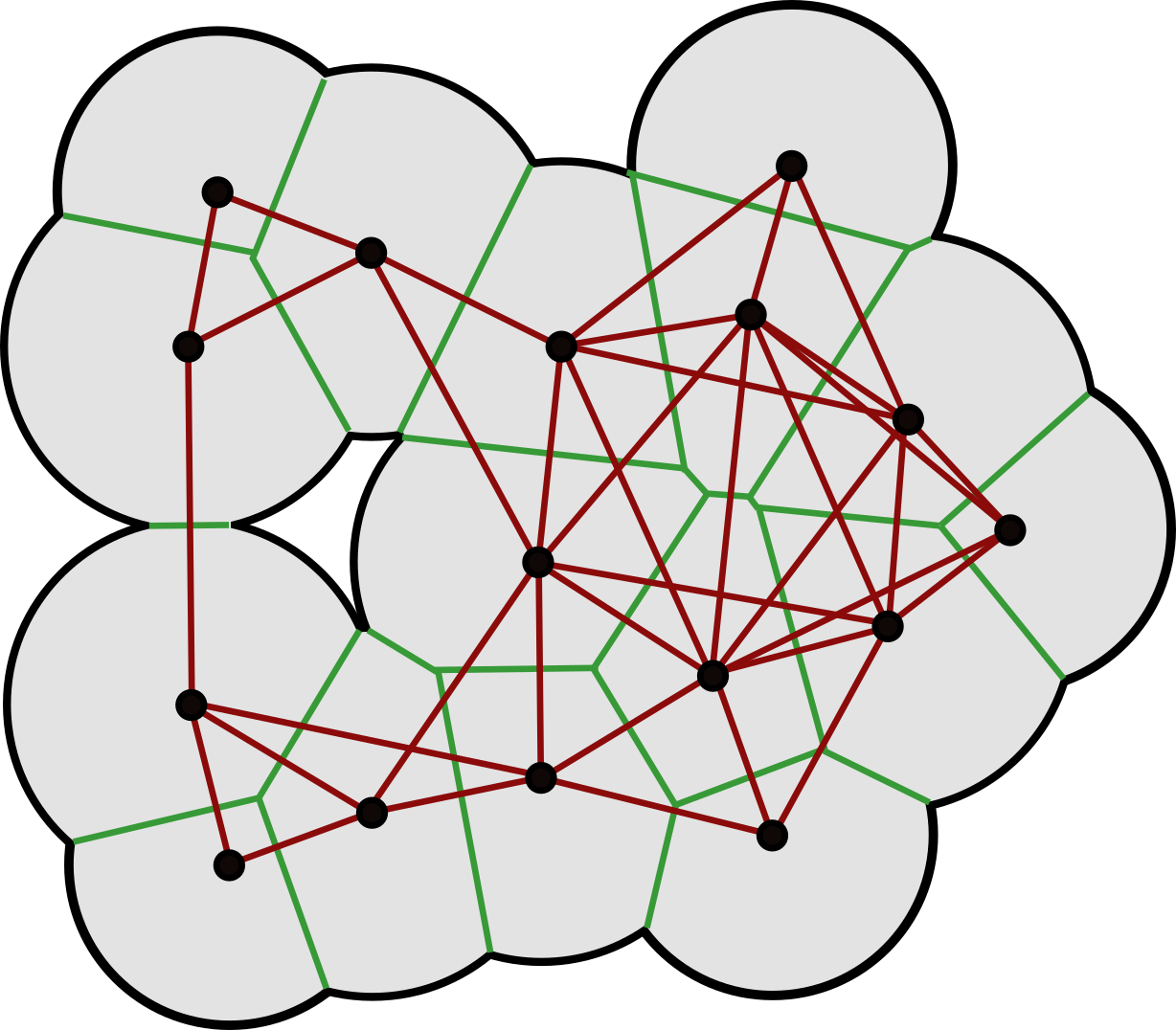}	
}	
\\%
\subcaptionbox{\label{fig:alphaTris}}{	
    \includegraphics[width=4cm]{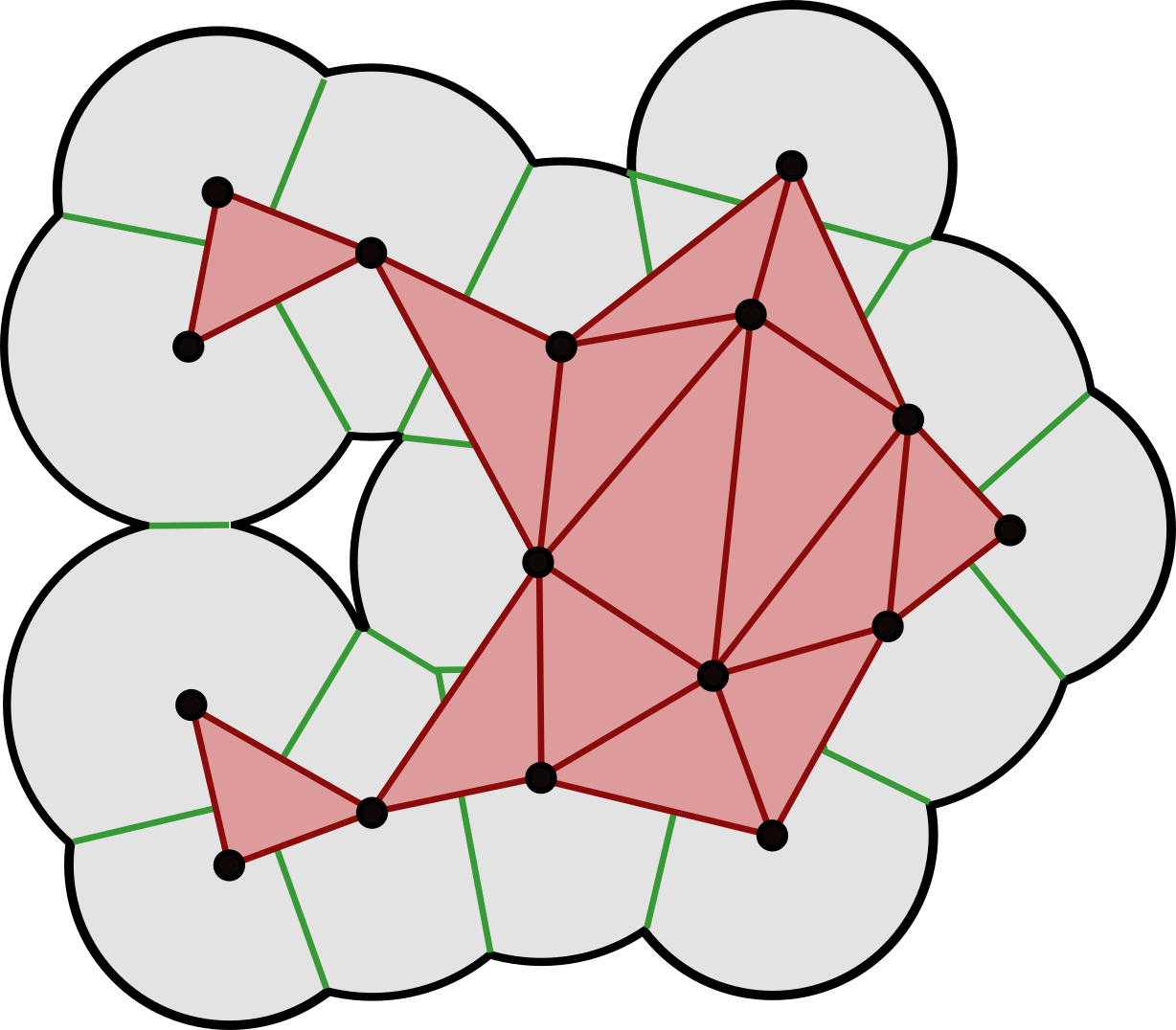}	
}	
\qquad	
\subcaptionbox{\label{fig:danglingEdges}}{	
    \includegraphics[width=4cm]{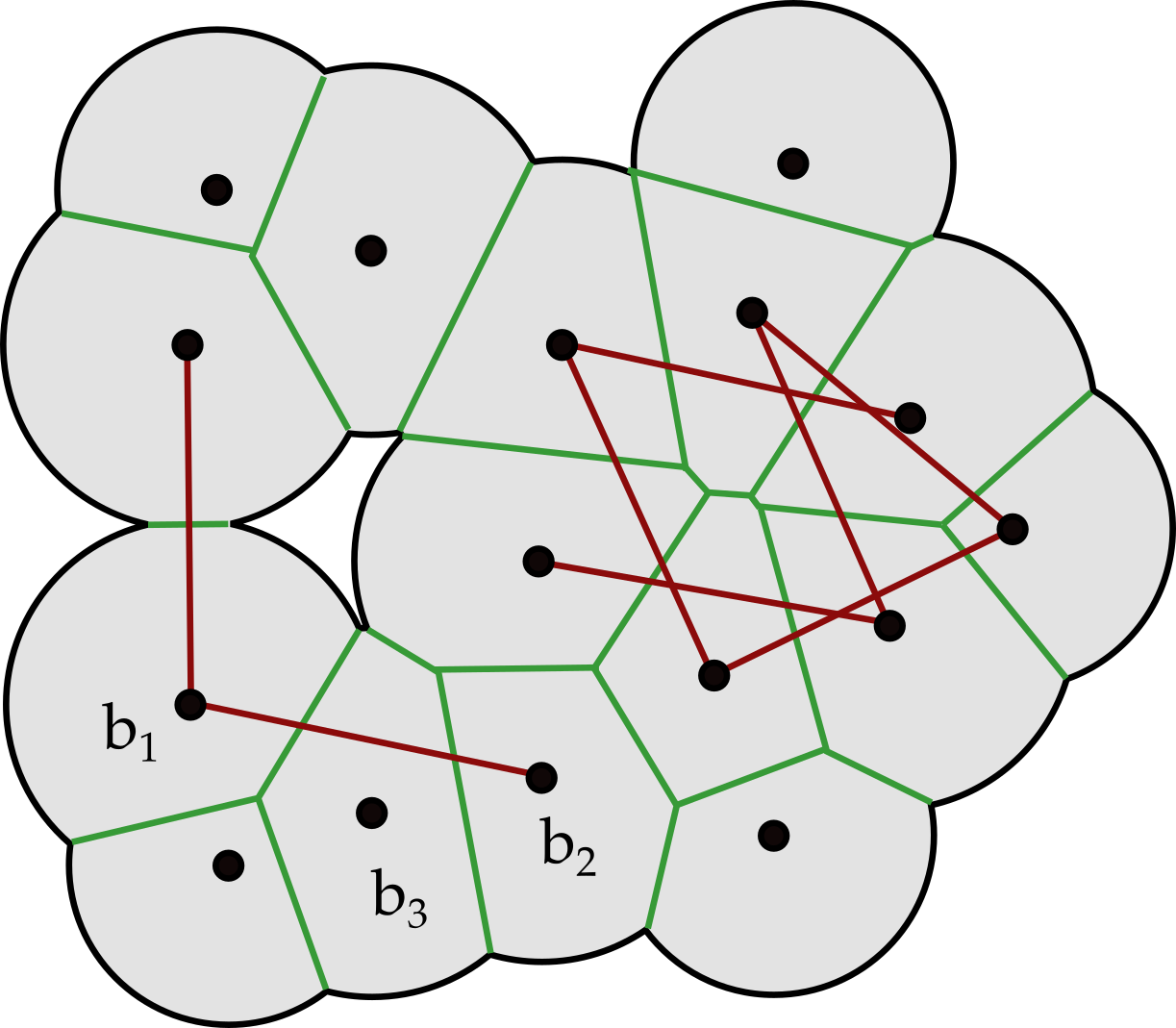}	
}	
\\%
\subcaptionbox{\label{fig:trueEdge}}{	
    \includegraphics[width=4cm]{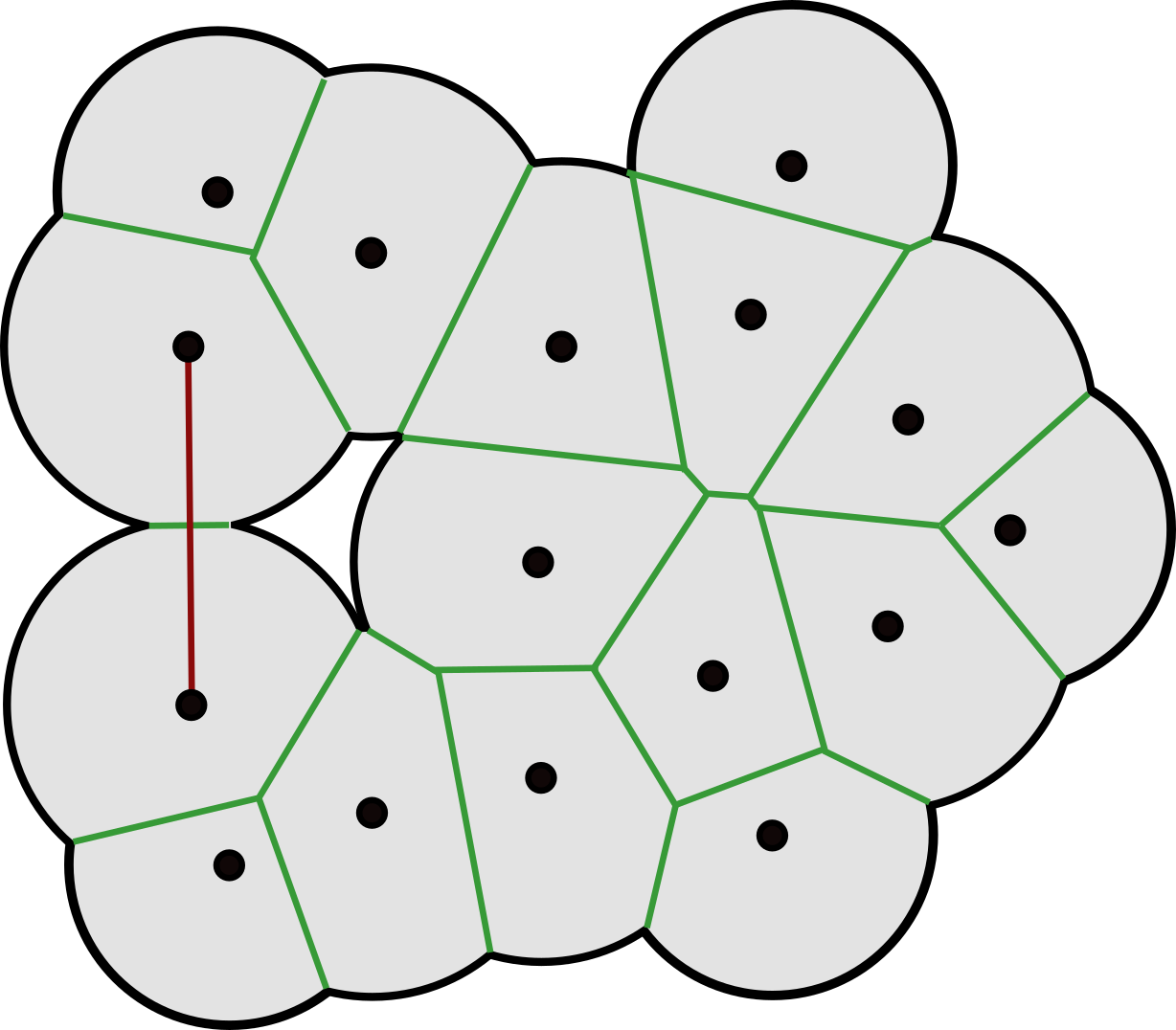}	
}	
\qquad    	
\subcaptionbox{\label{fig:finalAlpha}}{	
    \includegraphics[width=4cm]{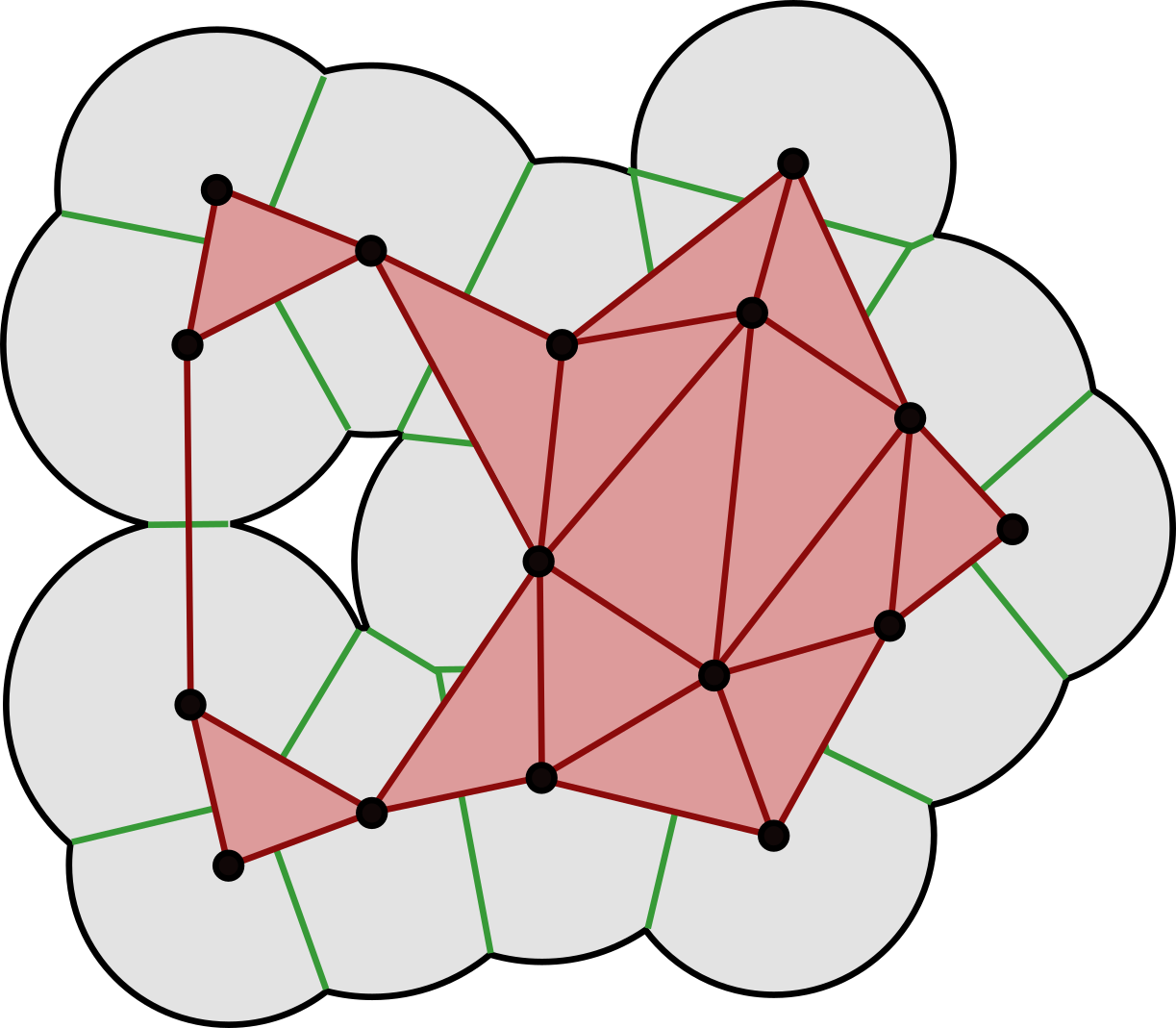}	
}
\caption[2D illustration of algorithm]{Illustration of the proposed algorithm in 2D.
\subref{fig:inputAlpha}~The set of disks $B$ grown by the parameter $\alpha$.
\subref{fig:initAlphaEdges}~First, compute the set of edges $\sigmao^1$ whose $\mathsf{OrthoSize} \leq \alpha$~(red). The triangles $\sigmao^2$ that satisfy this condition are also computed but they are not shown here.
\subref{fig:alphaTris}~Next, identify the triangles that satisfy AC2~(red).
\subref{fig:danglingEdges}~Collect edges in $\sigmao^1$ that are not incident on triangles in $K_\alpha^2$ into $\Sigma_{\mathsf{free}}^1$. Check if these edges satisfy AC2 with $p = p_{\mathsf{ortho}}^\sigma$. For example, the edge $b_1b_2$ does not satisfy this condition because $b_3$ is closer to $p_{\mathsf{ortho}}$ than $b_1$ and $b_2$.
\subref{fig:trueEdge}~One edge survives the AC2 check and thus belongs to $K_\alpha$. 
\subref{fig:finalAlpha}~The alpha complex is obtained as the union of $K_\alpha^2$, $K_\alpha^1$ and $K_\alpha^0$.
}
\label{fig:algoDemo}
\end{figure}

\subsection{Proof of correctness}
We now prove that that the algorithm described above correctly computes the alpha complex of the given set of weighted points by proving the following four claims. Each claim states that the set of simplices computed in Steps~2, 3, 4 and~5 are exactly the simplices belonging to the alpha complex. We assume that the input is non-degenerate. 
\begin{claim}
\bf Step~2 computes $K_\alpha^3$ correctly.
\end{claim}
\begin{proof}
For a tetrahedron $\sigma^3$, $p_{\mathsf{ortho}}^\sigma$ is the only point that satisfies condition AC1. In Step~2 of the proposed algorithm, we check if AC2 holds for $p_{\mathsf{ortho}}^\sigma$. If yes, then $p_{\mathsf{ortho}}^\sigma$ is a witness for $\sigma^3$, \emph{i.e.}, $p_{\mathsf{ortho}}^\sigma = p_{\mathsf{min}}^\sigma$. Further, since $\mathsf{OrthoSize}(\sigma^3)\leq \alpha$ and $p_{\mathsf{ortho}}^\sigma = p_{\mathsf{min}}^\sigma$, we have $\mathsf{Size}(\sigma^3)\leq \alpha$ thereby satisfying AC3. Therefore, $\sigma^3$ belongs to $K_\alpha^3$ because it satisfies all three conditions.
\end{proof}

We now prove that the algorithm correctly identifies the triangles of the alpha complex. 

\begin{lemma}
A triangle $\sigma^2 \in \Sigma_{\mathsf{free}}^2$ belongs to $K_\alpha^2$ if and only if it satisfies AC2 with $p = p_{\mathsf{ortho}}^\sigma$.
\label{lemma:danglingTri}
\end{lemma}

\begin{proof} 
We first prove the backward implication, namely if $\sigma^2 \in \Sigma_{\mathsf{free}}^2$ satisfies AC2 with $p = p_{\mathsf{ortho}}^\sigma$, then $\sigma^2 \in K_\alpha^2$. Note that $p_{\mathsf{ortho}}^\sigma$ satisfies AC1 by definition. Further, it satisfies AC2 by assumption and hence $\mathsf{Size}(\sigma^2) = \mathsf{OrthoSize}(\sigma^2)$. We also have $\mathsf{OrthoSize}(\sigma^2)\leq \alpha$ because $\sigma^1\in \Sigma_{\mathsf{free}}^2 \subseteq \sigmao^2$. So, $\mathsf{Size}(\sigma^2)\leq \alpha$ thereby satisfying AC3. The triangle $\sigma^2$ with $p = p_{\mathsf{ortho}}^\sigma$ satisfies all three conditions and hence belongs to $K_\alpha^2$.

\begin{figure}[ht!]
\centering
\includegraphics[width=10cm]{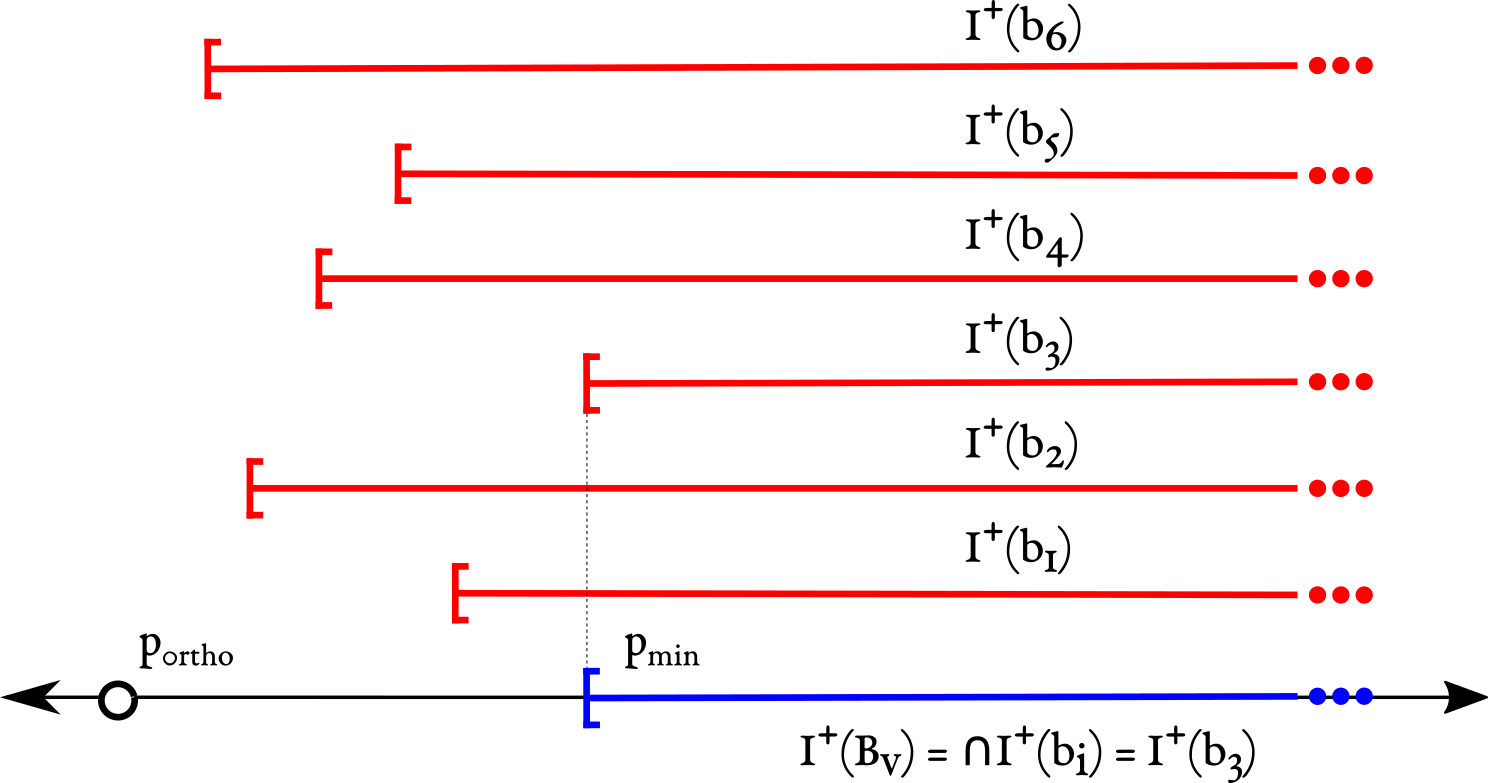}
\caption[The radical axis and half intervals induced by other balls]{
The radical axis of a triangle $\sigma^2$ is drawn such that $p_{\mathsf{ortho}}^\sigma$ is at the origin. A ball $b_i \in B - \sigma^2$ divides the radical axis into two half intervals. Points in the half interval $I^+(b_i)$ are closer to $b_0^\sigma$ as compared to $b_i$, \emph{i.e.} for all $p\in I^+(b_i)$, $\pi(p,b_0^\sigma)<\pi(p,b_i)$. Consider the set $B_v$ of balls that are closer to $p_{\mathsf{ortho}}^\sigma$ as compared to $b_0^\sigma$ So, $I^+(b_i)$ does not contain $p_{\mathsf{ortho}}^\sigma$. The intersection of these intervals, denoted by $I^+(B_v)$, is equal to one of the intervals $I^+(b_i)$. Here, $I^+(B_v)=I^+(b_3)$. The end point of the interval $I^+(b_3)$ is the closest witness for the tetrahedron $\sigma^2 \cup b_3$.
}
\label{fig:intervals}
\end{figure}

We will now prove the forward implication via contradiction. Suppose there exists a triangle $\sigma^2 \in \Sigma_{\mathsf{free}}^2$ that belongs to $K_\alpha^2$ but does not satisfy AC2 with $p = p_{\mathsf{ortho}}^\sigma$.  In other words, there exists a ball  $b_i \in B - \sigma^2$ for which $\pi(p_{\mathsf{ortho}}^\sigma, b_i) < \pi(p_{\mathsf{ortho}}^\sigma, b_0^\sigma)$. Let $B_v$ denote the set of all such balls $b_i$. The set of points that are equidistant from the three balls $(b_0^\sigma, b_1^\sigma, b_2^\sigma)$ corresponding to $\sigma^2$ form a line perpendicular to the plane containing $\sigma^2$ called the \emph{radical axis}. Each ball $b_i \in B_v$  partitions the radical axis into two half-intervals based on whether the point on radical axis is closer to $b_i$ or to $b_0^\sigma$, see Figure~\ref{fig:intervals}. Let $I^+(b_i)$ denote the half interval consisting of points that are closer to $b_0^\sigma$ compared to $b_i$. Let $I^+(B_v)$ denote the intersection of all such half intervals $I^+(b_i)$. We have assumed that $\sigma^2 \in K_\alpha^2$, so there must exist a closest witness $p_{\mathsf{min}}^\sigma$, and it lies within $I^+(B_v)$. Thus, $I^+(B_v)$ is non-empty. In fact, $I^+(B_v)=I^+(b_j)$ for some $ b_j\in B_v$ 
and $p_{\mathsf{min}}^\sigma$ is exactly the end point of $I^+(b_j)$. This implies that $p_{\mathsf{min}}^\sigma$ is also a closest witness for the tetrahedron $\sigma^3=(b_0^\sigma, b_1^\sigma, b_2^\sigma, b_j)$. So, $\sigma^3$ belongs to $K_\alpha^3$ and its $\mathsf{Size}$ is equal to $\mathsf{Size}(\sigma^2)$. However, this means that $\sigma^2 \notin \Sigma_{\mathsf{free}}^2$, a contradiction. So, the forward implication in the lemma is true.
\end{proof}

\begin{claim}
\bf Step~3 computes $K_\alpha^2$ correctly.
\end{claim}
\begin{proof}
If a simplex $\sigma^3$ belongs to $K_\alpha$ then naturally all of its faces also belong to $K_\alpha$. The algorithm includes such triangles into $K_\alpha^2$ and remove them from $\sigmao^2$ to obtain the set of free triangles $\Sigma_{\mathsf{free}}^2$ 
. It follows directly from Lemma~\ref{lemma:danglingTri} that AC2 is a necessary and sufficient condition for a triangle in $\Sigma_{\mathsf{free}}^2$  to belong to $K_\alpha^2$. Hence, Step~3 correctly computes the triangles belonging to $K_\alpha^2$. 
\end{proof}

\begin{figure}[ht!]
\centering
\includegraphics[width=10cm]{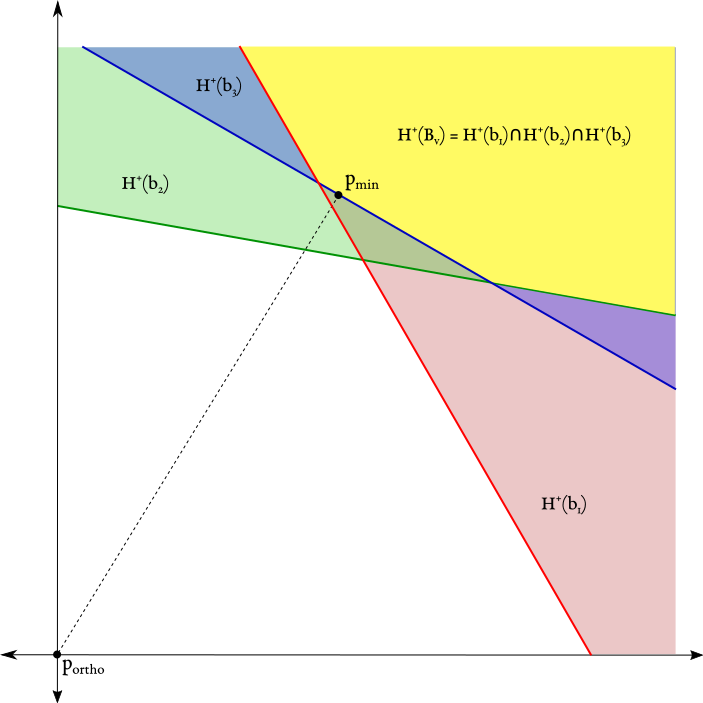}
\caption[The radical plane and half planes induced by other balls]{
The radical plane of an edge $\sigma^1$ is drawn such that $p_{\mathsf{ortho}}^\sigma$ is at the origin. A ball $b_i \in B - \sigma^1$ divides the radical plane into two half planes. The half plane $H^+(b_i)$ consists of points that are closer to $b_0^\sigma$ as compared to $b_i$, \emph{i.e.} for all $p\in H^+(b_i)$, $\pi(p,b_0^\sigma)<\pi(p,b_i)$. Let $B_v$ denote the set of balls that are closer to $p_{\mathsf{ortho}}^\sigma$ as compared to $b_0^\sigma$. The half planes $H^+(b_i)$ do not contain $p_{\mathsf{ortho}}^\sigma$. The intersection of these half planes, denoted by $H^+(B_v)$, is a convex region~(yellow). The power distance from $b_0^\sigma$ to a point $p\in H^+(B_v)$ is minimized at a point on the boundary of the convex region $H^+(B_v)$. But the boundary of $H^+(B_v)$ is a union of line segments that bound half planes $H^+(b_i)$. Here, the point at which the distance is minimum lies on the boundary of the half plane $H^+(b_3)$. This point is the closest witness for the triangle $\sigma^1 \cup b_3$.
}
\label{fig:halfPlanes}
\end{figure}
The above arguments need to be extended to prove that the edges of the alpha complex are also correctly identified. 

\begin{lemma} An edge $\sigma^1 \in \Sigma_{\mathsf{free}}^1$ belongs to $K_\alpha^1$ if and only if it satisfies the condition AC2 with $p = p_{\mathsf{ortho}}^\sigma$.
\label{lemma:danglingEdge}
\end{lemma}
\begin{proof} 
First, we assume that $\sigma^1 \in \Sigma_{\mathsf{free}}^1$ satisfies AC2 with $p = p_{\mathsf{ortho}}^\sigma$.
The point $p_{\mathsf{ortho}}^\sigma$ satisfies AC1 by definition. Further, it satisfies AC2 by assumption and hence $\mathsf{Size}(\sigma^1) = \mathsf{OrthoSize}(\sigma^1)$. We also have $\mathsf{OrthoSize}(\sigma^1)\leq \alpha$ because $\sigma^1 \in \Sigma_{\mathsf{free}}^1 \subseteq \sigmao^1$. So, $\mathsf{Size}(\sigma^1)\leq \alpha$ thereby satisfying AC3. The edge $\sigma^1$ with $p = p_{\mathsf{ortho}}^\sigma$ satisfies all three conditions and hence belongs to $K_\alpha^1$.

We will prove the forward implication via contradiction.  Suppose there exists an edge $\sigma^1 \in \Sigma_{\mathsf{free}}^1$ that belongs to $K_\alpha^1$ but does not satisfy AC2 with $p = p_{\mathsf{ortho}}^\sigma$.  In other words, there exists a ball  $b_i \in B - \sigma^1$ such that $\pi(p_{\mathsf{ortho}}^\sigma, b_i) < \pi(p_{\mathsf{ortho}}^\sigma, b_0^\sigma)$. Let $B_v$ denote the set of all such balls $b_i$. The set of points that are equidistant from the two balls $(b_0^\sigma, b_1^\sigma)$ corresponding to $\sigma^1$ form a plane perpendicular to the line containing $\sigma^1$ called the \emph{radical plane}. Each ball $b_i \in B_v$  partitions the radical plane into two half-planes based on whether the point on the radical plane is closer to $b_i$ or to $b_0^\sigma$, see Figure~\ref{fig:halfPlanes}. Let $H^+(b_i)$ denote the half-plane consisting of points that are closer to $b_0^\sigma$ compared to $b_i$. Let $H^+(B_v)$ denote the intersection of all such half-planes $H^+(b_i)$. We have assumed that $\sigma^1 \in K_\alpha^1$, so there must exist a closest witness $p_{\mathsf{min}}^\sigma$, and it lies within $H^+(B_v)$. Thus, $H^+(B_v)$ is non-empty. In fact, $p_{\mathsf{min}}^\sigma$ lies on the envelope of $H^+(B_v)$ because it minimizes the distance to $b_0^\sigma$. Let $p_{\mathsf{min}}^\sigma$ lie on the bounding line corresponding to $H^+(b_j)$ for some $ b_j\in B_v$. This implies that $p_{\mathsf{min}}^\sigma$ is also a closest witness for the triangle $\sigma^2=(b_0^\sigma, b_1^\sigma, b_j)$. So, $\sigma^2$ belongs to $K_\alpha^2$, and its size is equal to $\mathsf{Size}(\sigma^1)$. However, this means that $\sigma^1 \notin \Sigma_{\mathsf{free}}^1$, a contradiction. So, the forward implication in the lemma is also true.
\end{proof}

\begin{claim}
\bf Step~4 computes $K_\alpha^1$ correctly.
\end{claim}
\begin{proof}
All edge faces of triangles in $K_\alpha^2$ naturally belong to $K_\alpha^1$. Step~4 inserts all edges incident on triangles in $K_\alpha^2$ into $K_\alpha^1$ as valid edges and removes them from $\sigmao^1$ to obtain the set of free edges $\Sigma_{\mathsf{free}}^1$. It follows directly from Lemma~\ref{lemma:danglingEdge} that AC2 is  a necessary and sufficient condition for an edge $\sigma^1 \in \Sigma_{\mathsf{free}}^1$ to belong to $K_\alpha^1$. Therefore, Step~4 correctly computes the edges belonging to $K_\alpha^1$.
\end{proof}
\begin{claim}
\bf Step~5 computes $K_\alpha^0$ correctly.
\end{claim}
\begin{proof}
All vertices incident on $K_\alpha^1$ naturally belong to $K_\alpha^0$. Step~5 inserts all such vertices in $K_\alpha^0$ as valid vertices and removes them from $\sigmao^0$ to obtain the set of free vertices $\Sigma_{\mathsf{free}}^0$. Next, the vertices in $\Sigma_{\mathsf{free}}^0$ for which the center of the ball $b_i=(p_i,r_i)$ satisfies AC2 are also inserted into $K_\alpha^0$. Clearly, these vertices also satisfy AC3 because they belong to $\sigmao^0$. The condition AC1 is not relevant for $0$-dimensional simplices. Therefore, these vertices clearly belong to the alpha complex. Similar to Lemmas~\ref{lemma:danglingTri} and \ref{lemma:danglingEdge}, it is easy to prove that checking for AC2 for $p=p_i$ is necessary and sufficient condition to decide whether a vertex in $\Sigma_{\mathsf{free}}^0$ belongs to the alpha complex. That is, it is possible to show that vertices in alpha complex that have non-empty Voronoi regions but do not satisfy AC2 for $p=p_i$ would be incident on some edge in $K_\alpha^1$, and therefore must have been already detected by Step~4 and hence can not belong to $\Sigma_{\mathsf{free}}^0$. Therefore, Step~5 correctly computes the vertices belonging to $K_\alpha^0$.
\end{proof}

The arguments in the proofs of Lemma~\ref{lemma:danglingTri} and Lemma~\ref{lemma:danglingEdge} are similar. A general result is likely true for $d$-dimensional simplices in $\Sigma_{\mathsf{free}}^d$. However, given the focus on alpha complexes in $\mathbb{R}^3$, we prefer to state and prove these results specific to lower dimensions. We also prefer to provide individual proofs for edges and triangles because it simplifies the exposition and could also potentially help in the design of improved data structures to accelerate computation of different steps of the algorithm. 

\section{Parallel algorithm for biomolecules}
Although the algorithm as described above is provably correct, a straightforward implementation will be extremely inefficient with a worst-case running time of $O(n^5)$, where $n$ is the number of weighted points in $B$. This is because Step~1 requires $O(n^4)$ time to generate all possible tetrahedra. In later steps, we need $O(n)$ effort per simplex to check AC2. However, the input corresponds to atoms in a biomolecule. We show how certain properties of biomolecules can be leveraged to develop a fast parallel implementation. 

\subsection{Biomolecular data characteristics}
Atoms in a biomolecule are well distributed. The following three properties of biomolecules are most relevant: 
\begin{itemize}
\item The radius of an atom is bounded. The typical radius of an atom in a protein molecule ranges between 1\r{A} to 2\r{A}~\cite{bondi1964van}. Further, a protein molecule contains upwards of thousand atoms. So, the radius is small compared to the total size of the molecule. 
\item There is a lower bound on the distance between the centres of two atoms. This is called the \emph{van der Waals contact distance}, beyond which the two atoms start repelling each other. In the case of atoms in protein molecules, this distance is at least 1\r{A}. This property together with the upper bound on atomic radii ensures that no atom is completely contained inside another. This means that the weighted Voronoi regions corresponding to the atoms in a biomolecule can be always be assumed to be non-empty.
\item Structural biologists are interested in small values of $\alpha$. The two crucial values are 0\r{A} and 1.4\r{A}. The former corresponds to using van der Waals radius and the latter corresponds to the radius of water molecule, which acts as the solvent.
\end{itemize} 
In the light of the above three properties, we can say that the number of simplices of the alpha complex that are incident on a weighted point~(atom) is independent of the total number of input atoms and hence bounded by a constant~\cite{halperin1998spheres}. 

\subsection{Acceleration data structure}
The algorithm will benefit from an efficient method for accessing points of $B$ that belong to a local neighborhood of a given weighted point. We store the weighted points in a grid-based data structure. Let $r_{max}$ denote the radius of the largest atom and assume that the value of the parameter $\alpha$ is available as input. First, we construct a grid with cells of side length $\sqrt{r_{max}^2+\alpha}$ and then bin the input atoms into the grid cells. In our implementation, we do not store the grid explicitly because it may contain several empty cells. Instead, we compute the cell index for each input atom and sort the list of atoms by cell index to ensure that atoms that belong to a particular cell are stored at consecutive locations. The cell index is determined based on a row-major or column-major order. Alternatively, a space-filling curves like the Hilbert curve could also be used to order the cells. 

After the atoms are stored in grid cells, the alpha complex is computed in two stages. In the first stage, we employ a bottom-up approach to obtain a conservative estimate of the edges, triangles, and tetrahedra belonging to the alpha complex. The false positives from the first stage are removed in a subsequent pruning stage resulting in the correct alpha complex. We  describe these two stages in the following subsections.

\subsection{Potential simplices}
The first stage essentially corresponds to Step~1 of the algorithm described in the previous section. We compute the set $\sigmao$ of potential simplices for which $\mathsf{OrthoSize}(\sigma^d)\leq \alpha$. However, for efficiency reasons we process the simplices in the order of increasing dimension. First, we identify edges that satisfy the AC3 condition as described below. Given the size of the grid cell, endpoints of edges that satisfy the condition either lie within the same grid cell or in adjacent cells. So, the grid data structure substantially reduces the time required to compute the list of potential edges $\sigmao^1$. Beginning from this set of edges, we construct the set of all possible triangles and retain the triangles whose $\mathsf{OrthoSize}$ is no greater than $\alpha$, resulting in the set $\sigmao^2$.  Finally, we use the triangles in $\sigmao^2$ to construct the list of tetrahedra that satisfy the $\mathsf{OrthoSize} \leq \alpha$ condition. The above procedure works because the $\mathsf{OrthoSize}$ of a simplex is always greater than or equal to the $\mathsf{OrthoSize}$ of its faces. The set of simplices identified in this stage contains all simplices of the alpha complex. False positives are pruned in the second stage described below.

\subsection{Pruning}
The second stage corresponds to Steps~2-5 of the algorithm and processes the potential simplices in the decreasing order of dimension. This stage checks the characterizing condition AC2 to prune $\sigmao$ into $K_{\alpha}$. The tetrahedra are processed by checking if any of the input balls are closer to the $ortho\text{-}center$ than the balls incident on the tetrahedron. If yes, the tetrahedron is pruned away. Else, the tetrahedron is recognized as belonging to the alpha complex and inserted into $K_\alpha^3$. Triangles incident on these tetrahedra also belong to the alpha complex and are inserted into $K_\alpha^2$ after they are removed from the list of potential triangles $\sigmao^2$. Next, the triangles in $\sigmao^2$ are processed by checking if they satisfy AC2. If yes, they are inserted into $K_\alpha^2$. Otherwise, they are pruned away. All edges incident on triangles belonging to $K_\alpha^2$ are inserted into $K_\alpha^1$ and removed from the set $\sigmao^1$. Next, the edges in $\sigmao^1$ are processed by checking if they satisfy AC2. Edges that satisfy AC2 are inserted into $K_\alpha^1$ and the others are pruned away. All the vertices in $\sigmao^0$ are directly inserted into $K_\alpha^0$ without the AC2 check because for biomolecular data we assume that Voronoi regions of all the atoms are non-empty. The check for condition AC2 for each simplex is again made efficient by the use of the grid data structure. Atoms that may violate AC2 lie within the same cell as that containing the $ortho\text{-}center$ or within the adjacent cells. Atoms that lie within other cells may be safely ignored.

\subsection{CUDA implementation}
We use the CUDA framework~\cite{cudaWeb} and the \texttt{thrust} library~\cite{thrustWeb} within CUDA to develop a parallel implementation of the algorithm that executes on the many cores of the GPU. 
The grid computation is implemented as a CUDA kernel where all atoms are processed in parallel. The computation of potential simplices and pruning stages are broken down into multiple CUDA kernels and parallelized differently in order to increase efficiency. We now describe the parallelization strategy in brief. 

For computing the set of potential edges, the initial enumeration of possible edges incident on an atom is done using the atoms in the corresponding grid cell and its neighbouring cells. This is done per atom in parallel, the thread corresponding to the atom $i$ being responsible for generating the edges $ij, j > i$ to ensure no duplicate edges are generated. Subsequently, the AC3 condition is checked for the edges in parallel to finally generate the list of potential edges $\sigmao^1$. For computing potential triangles $\sigmao^2$, the potential edge list is used as a starting point for the initial enumeration of all possible triangles. This step is also parallelized per atom, the thread $i$ being responsible for generation of triangles of the type $ijk; j,k > i$ if all three edges $ij$, $ik$ and $jk$ are potential edges. The AC3 condition for the triangle is checked next within a separate kernel and parallelized per triangle to generate the potential triangles $\sigmao^2$. A similar strategy is used for computing the set of potential tetrahedra $\sigmao^3$. 

The pruning stage is parallelized per tetrahedron, triangle, and edge as required. So, computation of $p_{\mathsf{ortho}}^\sigma$ is done in parallel. The grid data structure is again useful in checking for potential violators of the AC2 condition. Only the balls belonging to the grid cell corresponding to $p_{\mathsf{ortho}}^\sigma$ or the those in neighbouring grid cells can violate the AC2 condition. 

\subsection{Handling large data sizes}
Typical protein structures consist of up to 100,000 atoms. Our implementation can handle datasets of this size easily for reasonable values of $\alpha$. However, the size of datasets is ever increasing. Protein complexes that are available nowadays may consist of millions of atoms, necessitating smart management of GPU memory while handling such data sets.

We propose two strategies and implement one of them. The first strategy is to partition the grid by constructing an octree data structure and choosing an appropriate level in the octree to create partitions. Each partition together with its border cells can be processed independently of other partitions. So, we can copy one partition and its border to the GPU memory, compute its alpha complex, and copy the results back from GPU to CPU memory. After all the partitions are processed, the list of simplices can be concatenated followed by duplicate removal to generate the final alpha complex. 

The second strategy is to partition the sorted list of atoms into \emph{chunks} of equal sizes and to process each chunk independently. Here, we assume that the complete list of atoms together with the grid data structure fits in the GPU memory. This is a reasonable assumption considering that datasets containing several million atoms can easily fit on modern GPUs, which typically have at least 2GB video memory. Also, the main difficulty in handling large protein structures is managing the large lists of simplices generated within the intermediate steps of the algorithm, when compared to handling the input list of atoms or the output list of simplices. We compute the alpha complex by executing the algorithm in multiple passes. Each pass computes the alpha complex for a single chunk and copies it back to the CPU memory. We have implemented this second strategy and can handle data sizes of up to 16 million atoms on a GPU with 2GB of memory. Results are reported in the next section.

\section{Experimental results}
We now present results of computational experiments, which demonstrate that the parallel algorithm is fast in practice and significantly better than the state-of-the-art. We also performed runtime profiling to better understand the bottlenecks and effect of the parameter $\alpha$ on the runtime. 
We present results for $\alpha$ in the range $0.0$ to $2.0$. This range is important for structural analysis of biomolecules as it corresponds to \emph{solvent accessible surface} of the biomolecule for typical solvent molecules like water~(van der Waals radius = 1.4\r{A}). The value $\alpha = 0$ corresponds to the \emph{van der Waals} surface of the biomolecule.
All experiments, unless stated otherwise, were performed on a Linux system with an nVidia GTX~660~Ti graphics card running CUDA~8.0 and a 2.0GHz Intel Xeon octa core processor with 16 GB of main memory. The default number of threads per block was set at $512$ for all the CUDA kernels.

Mach and Koehl describe two techniques for computing alpha complex of biomolecules called \emph{AlphaVol} and \emph{UnionBall} in their paper~\cite{koehl2011volumes}. Both approaches construct the weighted Delaunay triangulation of input atoms first followed by a filtering step to obtain the alpha complex.\emph{UnionBall} is the state-of-the-art technique for alpha complex computation for biomolecules on multi-core CPU. It uses heuristics and optimizations specific to biomolecular data to improve upon \emph{AlphaVol}. For biomolecules containing 5~million atoms, \emph{AlphaVol} takes approximately 8600 seconds for computing the alpha complex, while \emph{UnionBall} takes approximately 150 seconds. Our method computes the alpha complex in less than 3 seconds for similar sized data, see Table~\ref{tab:compare660}. 

\subsection{Comparison with \emph{gReg3D}} 
\begin{table}[ht!]
\centering
\caption{Runtime comparison of the proposed algorithm with \emph{gReg3D} on an nVidia GTX 660 Ti graphics card. Timings are reported in milliseconds. $\%$Simplex refers to the size of the alpha complex as a percentage of the size of the weighted Delaunay triangulation. The last column shows the speedup in runtime of our algorithm over \emph{gReg3D}. `*' indicates the data was partitioned and processed in chunks. `--' indicates that the code could not execute due to insufficient memory.}
\resizebox{\columnwidth}{!}{
\begin{tabular}{ @{}ccrrrrrrr@{} }
\toprule
\multirow{2}{*}{$\alpha$} & \multirow{2}{*}{PDB id} & \multirow{2}{*}{$\#$Atoms} & \multicolumn{2}{c}{$K_{\alpha}$} & \multicolumn{2}{c}{\emph{gReg3D}} & \multirow{2}{*}{$\%$Simplex} & \multirow{2}{*}{Speed up}\\
 \cmidrule{4-5} \cmidrule{6-7}
 & & & $\#$Simplices & Time(ms) & $\#$Simplices & Time(ms) & &  \\
\toprule
\multirow{12}{*}{0.0} 
& 1GRM & 260 & 932 & 13 & 6295 & 117 & 14.8 & \textbf{9.0}\\
& 1U71 & 1505 & 5696 & 13 & 40878 & 115 & 13.9 & \textbf{11.1}\\
& 3N0H & 1509 & 5739 & 14 & 41244 & 137 & 13.9 & \textbf{10.0}\\
& 4HHB & 4384 & 38796 & 29 & 150141 & 193 & 25.8 & \textbf{6.6}\\
& 2J1N & 8142 & 29642 & 18 & 227719 & 229 & 13.0 & \textbf{12.7}\\
& 1K4C & 16068 & 62851 & 27 & 446383 & 347 & 14.1 & \textbf{12.9}\\
& 2OAU & 16647 & 123175 & 56 & 466586 & 344 & 26.4 & \textbf{6.2}\\
& 1AON & 58674 & 262244 & 65 & 1650841 & 879 & 15.9 & \textbf{13.5}\\
& 1X9P* & 217920 & 924086 & 113 & 6142811 & 2555 & 15.0 & \textbf{22.6}\\
& 1IHM* & 677040 & 2713083 & 277 & -- & -- & -- & \textbf{--}\\
& 4CWU* & 5905140 & 23450403 & 2709 & -- & -- & -- & \textbf{--}\\
& 3IYN* & 5975700 & 24188892 & 2874 & -- & -- & -- & \textbf{--}\\
\midrule
\multirow{12}{*}{1.0} 
& 1GRM & 260 & 1598 & 15 & 6295 & 117 & 25.4 & \textbf{7.9}\\
& 1U71 & 1505 & 10828 & 17 & 40878 & 115 & 26.5 & \textbf{8.5}\\
& 3N0H & 1509 & 10965 & 30 & 41244  & 137 & 26.6 & \textbf{4.6}\\
& 4HHB & 4384 & 65987 & 86 & 150141 & 193 & 44.0 & \textbf{2.2}\\
& 2J1N & 8142 & 58205 & 30 & 227719 & 229 & 25.6 & \textbf{7.6}\\
& 1K4C & 16068 & 118467 & 52 & 446383 & 347 & 26.5 & \textbf{6.7}\\
& 2OAU & 16647 & 199101 & 159 & 466586 & 344 & 42.7 & \textbf{2.2}\\
& 1AON & 58674 & 495683 & 160 & 1650841  & 879 & 30.0 & \textbf{5.5}\\
& 1X9P* & 217920 & 1653778 & 196 & 6142811  & 2555 & 26.9 & \textbf{13.0}\\
& 1IHM* & 677040 & 5058507 & 605 & -- & -- & -- & \textbf{--}\\
& 4CWU* & 5905140 & 44411353 & 5118 & -- & -- & -- & \textbf{--}\\
& 3IYN* & 5975700 & 45790463 & 5501 & -- & -- & -- & \textbf{--}\\
\bottomrule
\end{tabular}
}
\label{tab:compare660}
\end{table}
We are not aware of any available software that can compute the alpha complex directly without first constructing the complete Delaunay triangulation. In order to compare the performance, we chose the state-of-the-art parallel algorithm for computing the weighted Delaunay triangulation in 3D, \emph{gReg3D}~\cite{cao2014gpu}. The CUDA implementation of \emph{gReg3D} is available in the public domain. Table~\ref{tab:compare660} compares the running times of our proposed algorithm with that of \emph{gReg3D} for twelve different biomolecules at $\alpha=0$ and $\alpha=1$. As evident from the table, we consistently observe significant speedup over \emph{gReg3D}. The observed speedup is as high as a factor of 22 for the biomolecule 1X9P at $\alpha=0$, one of the largest molecules in our dataset. Clearly, the speedup goes down for $\alpha=1$ when compared to $\alpha=0$ because of the increased number of simplices in the output alpha complex. We also report the number of simplices in the alpha complex compared to the total number of simplices in the Delaunay triangulation under the column `$\%$Simplex'. This makes it clear why the speedup decreases as $\alpha$ is increased from 0 to 1. For example, for the protein 1AON, the fraction of alpha complex simplices increases from $15.9\%$ to $30\%$ as $\alpha$ is increased from 0 to 1. Correspondingly, the speedup decreases from a factor of 13.5 to 5.5. We  repeated the experiment on a MS~Windows system with an nVidia GTX~980~Ti card running CUDA~8.0 and observed similar speedups. However, the individual runtimes both for our algorithm and for \emph{gReg3D} were higher on the GTX~980~Ti. 

The starred entries in Table~\ref{tab:compare660} are results for execution using the data partitioning approach. This is necessitated because these four large molecules generate large intermediate simplex lists that can not fit into the GPU memory if all the atoms in the molecule are processed at once. We observe that \emph{gReg3D} is able to successfully compute the Delaunay complex for only one out of these four large molecules and runs out of GPU memory for the remaining three molecules. 

\subsection{Runtime profiling} 
\begin{table}[b!]
\centering
\caption{Time spent within different steps of the algorithm. Timings are reported in milliseconds for memory transfer, grid computation, computing potential simplices, and pruning. The last column shows the total time taken for all steps.}
\resizebox{\columnwidth}{!}{
\begin{tabular}{ @{}crrrrrrrrrrrr@{} }
\toprule
\multirow{2}{*}{$\alpha$} & \multirow{2}{*}{PDB id} & \multirow{2}{*}{$\#$Atoms} & \multirow{2}{*}{$\#$Simplices} &\multirow{2}{*}{Memory} & \multirow{2}{*}{Grid} & \multicolumn{3}{c}{Potential Simplices} & \multicolumn{3}{c}{Pruning} & \multicolumn{1}{c}{Total}\\
 \cmidrule{7-9} \cmidrule{10-12}
 & & & & & & Edges & Tris & Tets & Tets & Tris & Edges & \multicolumn{1}{c}{time}\\
\toprule
\multirow{8}{*}{0.0} 
& 1GRM & 260 & 932 & 0.8 & 1.0 & 2.8 & 1.1 & 1.0 & 1.2 & 3.1 & 1.9 & 13.0\\
& 1U71 & 1505 & 5696 & 0.9 & 0.8 & 2.3 & 1.7 & 1.0 & 1.8 & 2.6 & 1.9 & 13.0\\
& 3N0H & 1509 & 5739 & 0.7 & 0.9 & 2.3 & 1.4 & 2.5 & 1.8 & 2.5 & 1.5 & 13.7\\
& 4HHB & 4384 & 38796 & 1.1 & 0.8 & 2.7 & 2.0 & 6.0 & 8.3 & 4.6 & 3.7 & 29.2\\
& 2J1N & 8142 & 29642 & 1.1 & 1.2 & 3.7 & 1.6 & 1.3 & 2.0 & 3.9 & 3.2 & 18.1\\
& 1K4C & 16068 & 62851 & 1.7 & 2.0 & 4.3 & 1.5 & 1.6 & 3.8 & 7.6 & 4.3 & 26.9\\
& 2OAU & 16647 & 123175 & 2.1 & 1.3 & 4.7 & 4.3 & 5.8 & 21.9 & 9.5 & 6.0 & 55.5\\
& 1AON & 58674 & 262244 & 4.6 & 2.8 & 11.1 & 5.6 & 4.1 & 16.7 & 10.9 & 9.4 & 65.2\\
\midrule
\multirow{8}{*}{1.0} 
& 1GRM & 260 & 1598 & 1.2 & 1.4 & 2.7 & 1.2 & 2.0 & 1.8 & 2.7 & 1.8 & 14.7\\
& 1U71 & 1505 & 10828 & 0.9 & 0.8 & 2.3 & 1.5 & 3.4 & 2.5 & 3.2 & 2.5 & 17.1\\
& 3N0H & 1509 & 10965 & 0.9 & 1.4 & 2.4 & 1.7 & 14.6 & 2.6 & 3.5 & 2.9 & 30.0\\
& 4HHB & 4384 & 65987 & 1.5 & 1.9 & 3.4 & 5.4 & 23.7 & 32.8 & 12.4 & 4.8 & 86.0\\
& 2J1N & 8142 & 58205 & 1.4 & 1.1 & 3.4 & 2.5 & 3.6 & 9.0 & 4.6 & 4.5 & 30.3\\
& 1K4C & 16068 & 118467 & 2.1 & 1.8 & 6.1 & 3.0 & 3.4 & 19.5 & 8.3 & 7.9 & 52.2\\
& 2OAU & 16647 & 199101 & 3.0 & 1.7 & 6.0 & 10.2 & 28.3 & 90.0 & 12.1 & 7.5 & 158.9\\
& 1AON & 58674 & 495683 & 6.3 & 2.0 & 12.4 & 9.9 & 12.5 & 87.9 & 17.9 & 11.0 & 159.9\\
\bottomrule
\end{tabular}
}
\label{tab:timeSplit}
\end{table}
\begin{figure}[ht!]
\centering
\subcaptionbox{$\alpha=0.0$\label{fig:timeAll1}}	
    {\includegraphics[width=11cm]{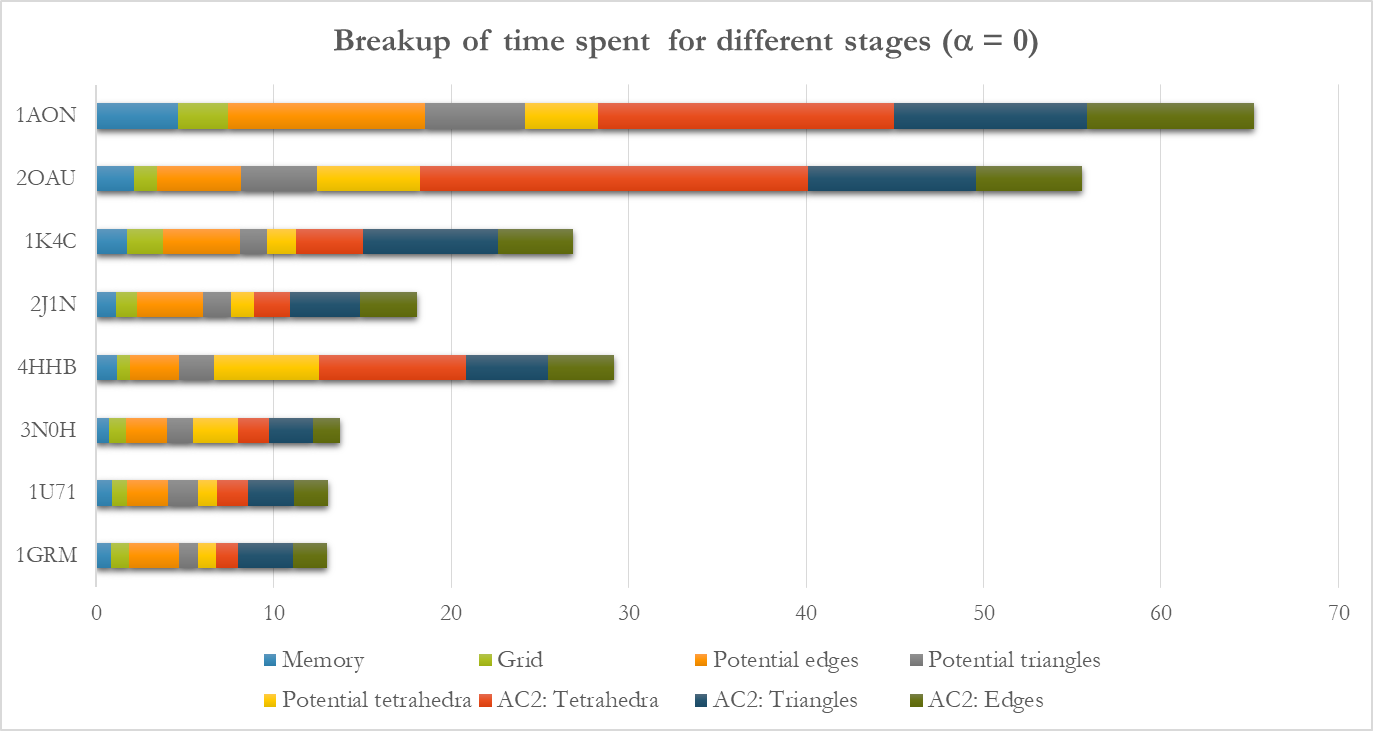}}	
    \qquad	
\subcaptionbox{$\alpha=1.0$\label{fig:timeAll2}}	
    {\includegraphics[width=11cm]{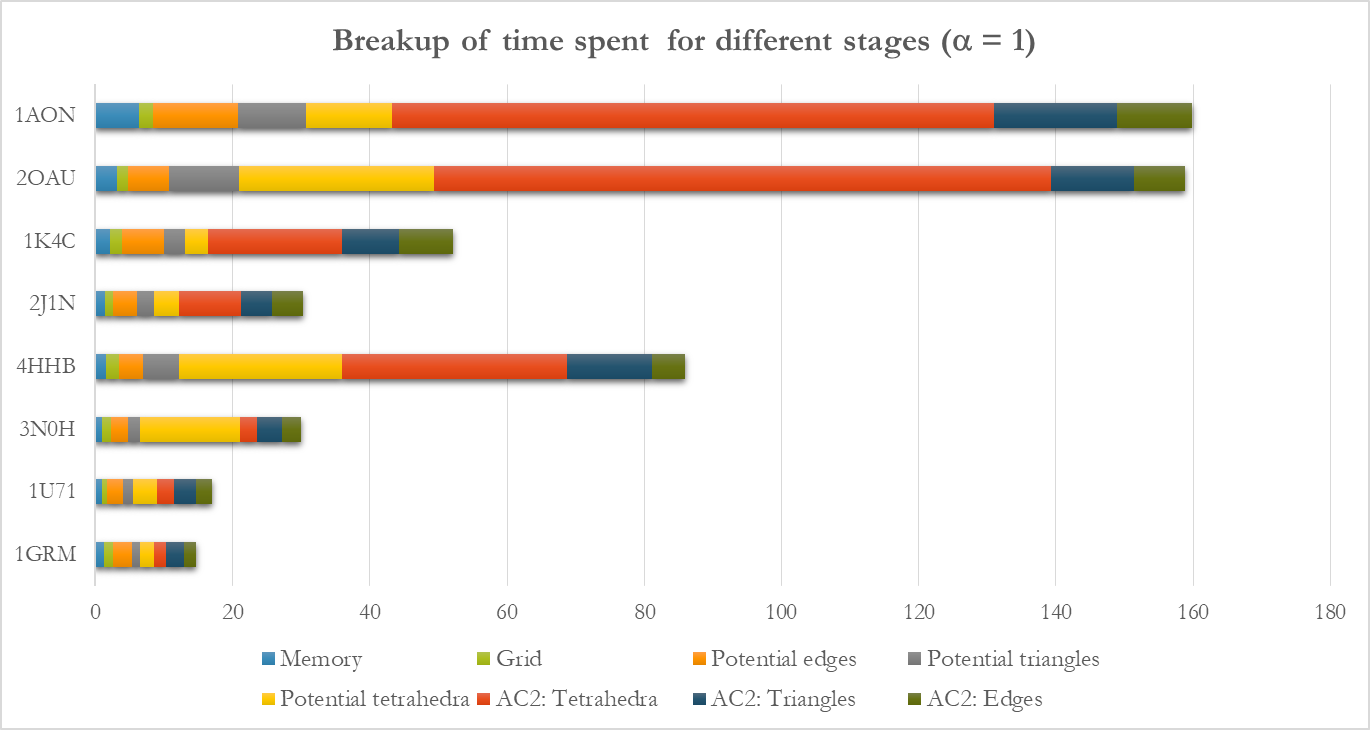}}
\caption[Time spent at different steps]{Time spent for different steps of the algorithm. 
}
\label{fig:timeAll}
\end{figure}
\begin{figure}[ht!]
\centering
\subcaptionbox{$\alpha=0.0$ \label{fig:timeSplitAll1}}	
    {\includegraphics[width=11cm]{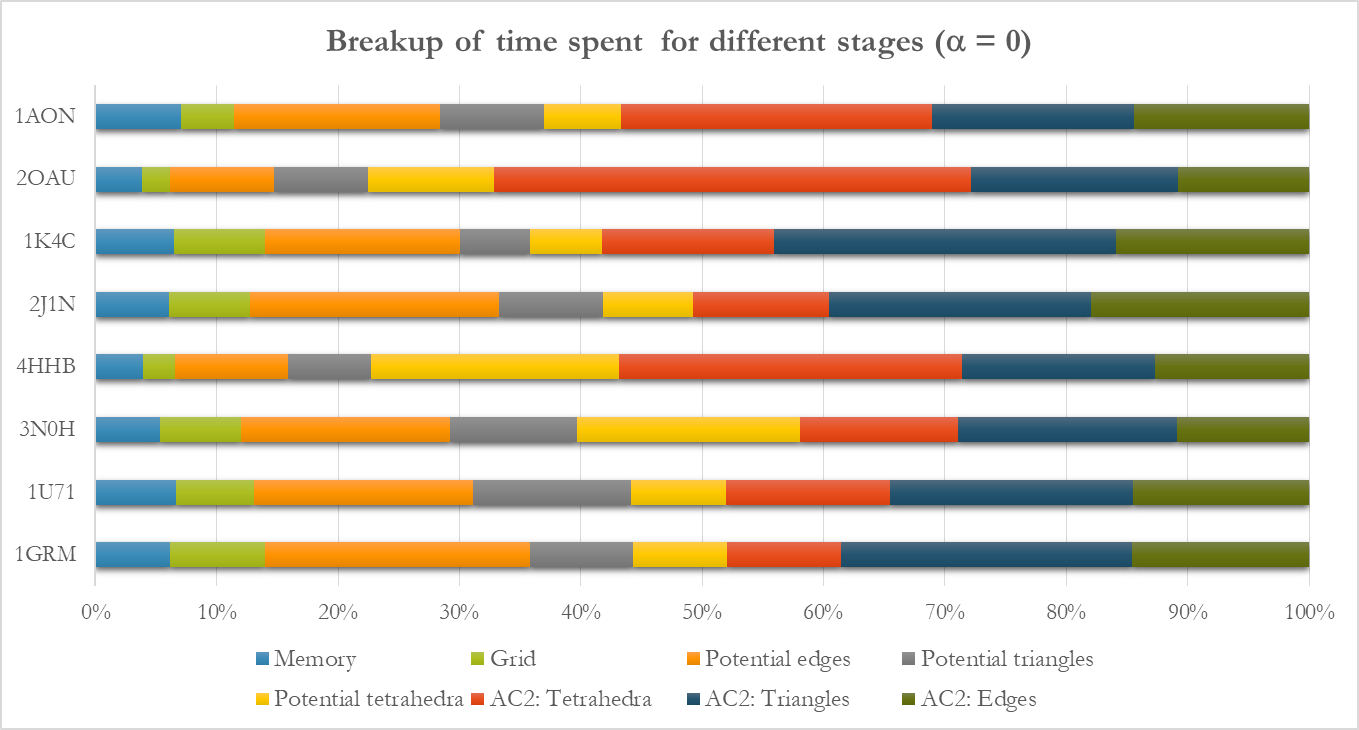}}	
    \qquad	
\subcaptionbox{$\alpha=1.0$ \label{fig:timeSplitAll2}}	
    {\includegraphics[width=11cm]{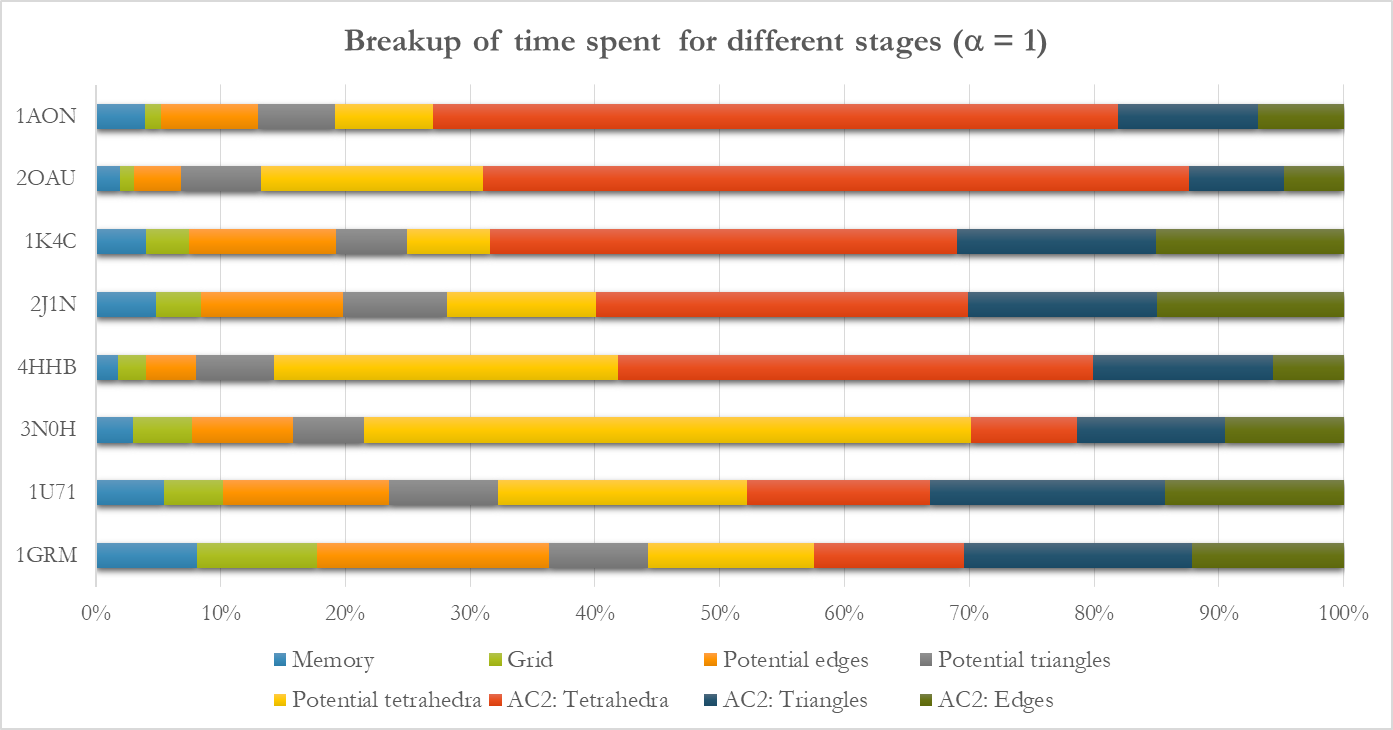}}
\caption[Proportion of time spent at different steps]{Proportion of time spent for different steps during the execution of the algorithm. The pruning stage~(AC2 tetrahedra, triangles and edges checks) takes significantly more effort compared to other steps. Also, the time spent for this step increases with $\alpha$.
}
\label{fig:timeSplitAll}
\end{figure}
\begin{figure}[ht!]
\centering
\includegraphics[width=10cm]{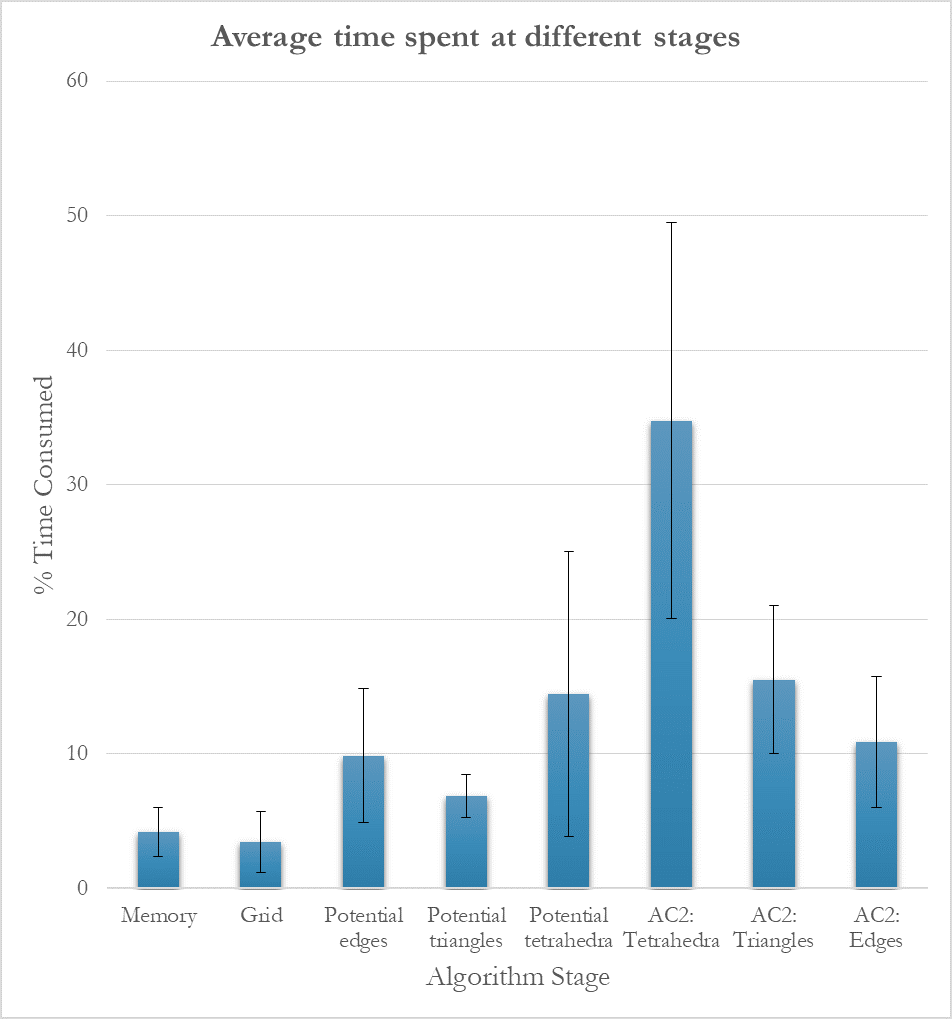}
\caption[Split up of time spent at different steps]{The average proportion of the execution time accounted for by various steps of computation. This average was taken over all the molecules in our dataset at various values of $\alpha$ varying from $0$ to $2.0$. As evident from the error bars, there is significant variability. But, in general the pruning stage of the algorithm, specially the tetrahedra computation step takes the maximum time.
}
\label{fig:timeSplitError}
\end{figure}
The two stages of our parallel algorithm (potential simplices and pruning) are further divided into three steps each, corresponding to the computation of edges, triangles, and tetrahedra respectively. A grid computation step precedes the two stages. We study the computation effort for each of these seven steps of the algorithm. We also report the time spent in memory transfers from CPU to GPU and vice-versa. Thus, we report the split up of the total runtime into eight categories namely, `Memory transfer', `Grid computation', `Potential edges', `Potential triangles', `Potential tetrahedra', `AC2 tetrahedra', `AC2 triangles' and `AC2 edges'. 

Table~\ref{tab:timeSplit} summarizes the observed split up of runtime for eight different biomolecules. Figure~\ref{fig:timeAll} shows the actual time spent for different steps and Figure~\ref{fig:timeSplitAll} shows relative time spent for each step. From these figures, it is clear that the pruning stage consumes the maximum amount of time. The pruning stage involves checking the neighboring balls for violations of the AC2 condition for each simplex. Specially, the tetrahedra-pruning step (red) takes approximately $25\%$ of the total time required for alpha complex computation.  

We  performed additional experiments to determine the average split up over multiple runs. We computed the relative time spent for each step for different values of $\alpha$ between 0.0 and 2.0. These observations are reported in Figure~\ref{fig:timeSplitError}. It is clear that the memory transfers and grid computation combined do not take more than $10\%$ of the total time. The potential-simplices-estimation stage consumes $30\%$ of the time. However, the pruning stage is most expensive, taking up roughly $60\%$ of the computation time. The pruning-tetrahedra step takes up $35\%$ of the time on average. This suggests that this step should be the focus of the optimization efforts in future. It should be noted that proportion of time spent for each step depends on the distribution of atoms in the biomolecule as well as the value of $\alpha$. This explains the significant deviation from the averages as shown by the error bars.

\subsection{Effect of the value of alpha} 
\begin{figure}[b!]
\centering
\includegraphics[width=12cm]{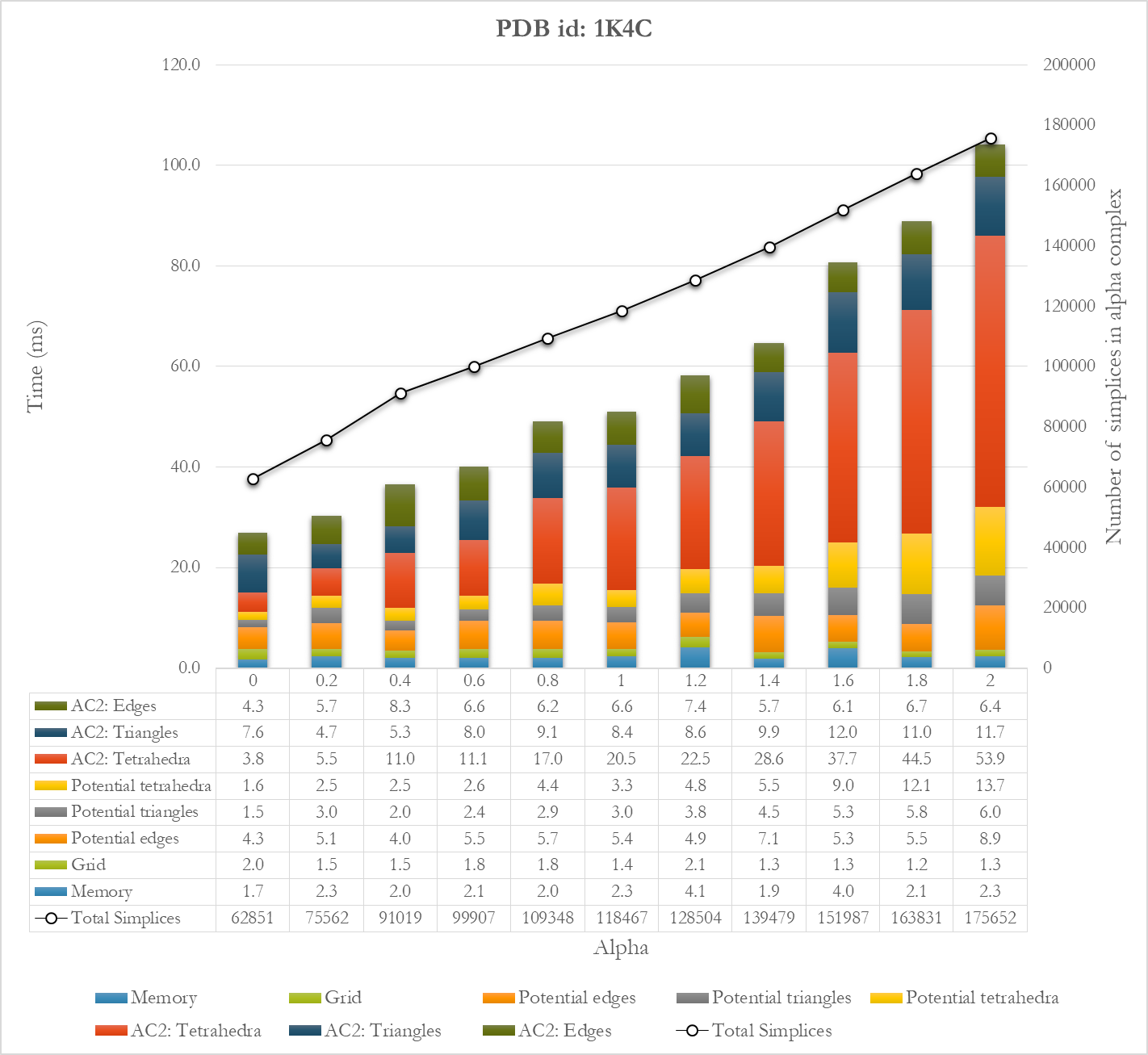}
\caption[Running time details for 1K4C]{Running time for varying values of $\alpha$ for 1K4C. The number of simplices in the output alpha complex is also shown (black line). The number of simplices increases almost linearly with $\alpha$ as expected from the distribution of atoms in typical biomolecules. The running time also increases almost linearly with $\alpha$ for this molecule. Also, the fraction of time spent for tetrahedra computation step~(red) increases with $\alpha$.
}
\label{fig:1K4C_times}
\end{figure}
\begin{figure}[ht!]
\centering
\includegraphics[width=12cm]{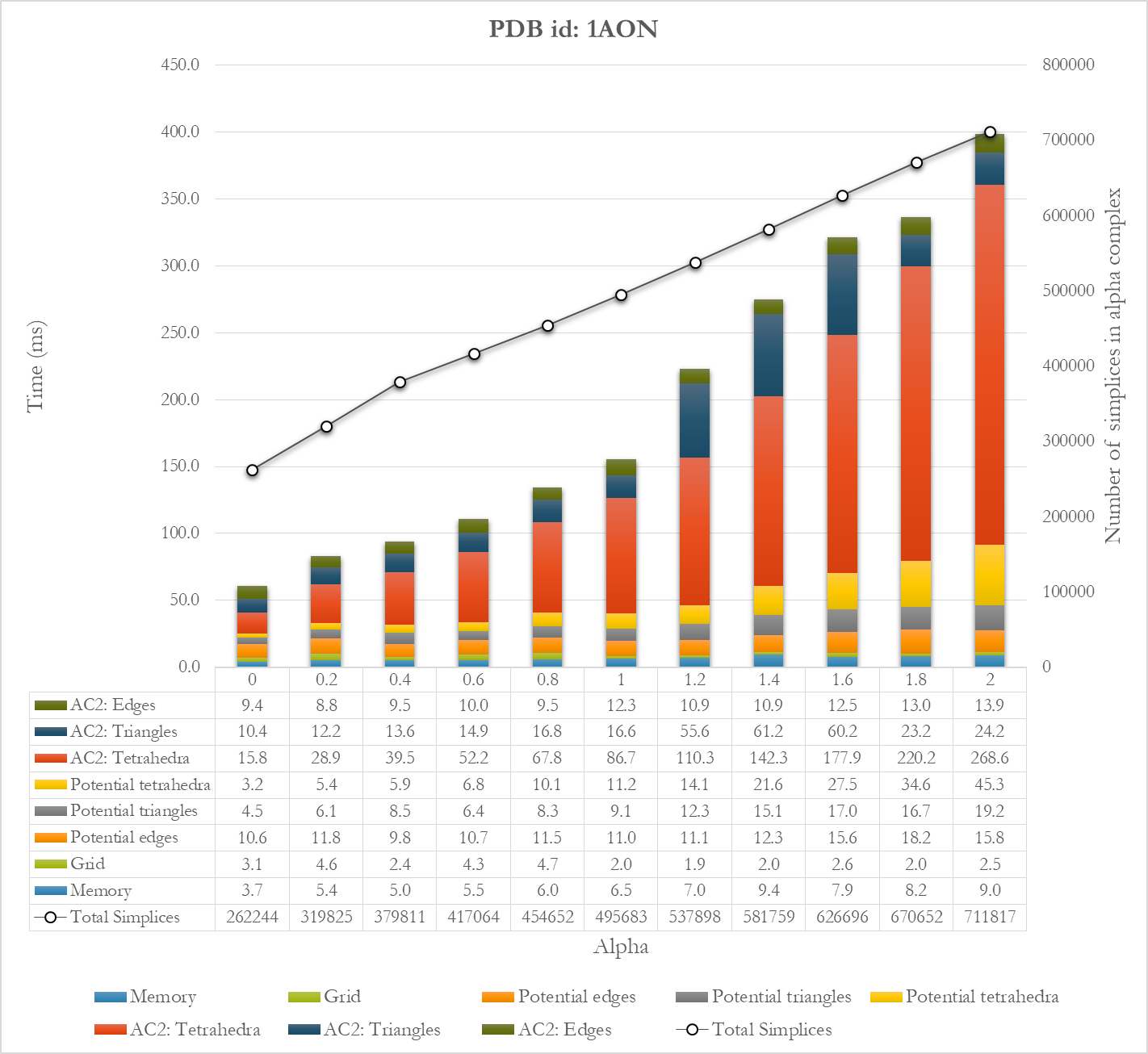}
\caption[Running time details for 1AON]{Running time for varying values of $\alpha$ for 1AON. The number of simplices in the output alpha complex is also shown (black line). The running time increases almost linearly with $\alpha$ for this molecule.
}
\label{fig:1AON_times}
\end{figure}
We also performed experiments to observe the runtime performance as the value of $\alpha$ is varied between 0.0 and 2.0. Figures~\ref{fig:1K4C_times} and~\ref{fig:1AON_times} show the results for the proteins 1K4C and 1AON, respectively. We also show how the number of simplices in the computed alpha complex increases as the value of $\alpha$ is increased. The runtime and total number of  simplices follow a near-linear trend. However, increase in time required for pruning, especially for pruning the tetrahedra, is greater than time required for other steps of the algorithm. Note that although both graphs appear linear, this is not guaranteed behavior for other input. The scaling behavior depends on the distribution of the atoms in the molecule and on the range of $\alpha$ values for which the experiment is conducted.

\subsection{Numerical issues}
\begin{table}[t!]
\centering
\caption{Incorrectly identified simplices of the alpha complex.}
\begin{tabular}{ @{}crrrrrrrr@{} }
\toprule
 $\alpha$ & PDB id & $\#$Atoms & $\#$Simplices & \multicolumn{4}{c}{$\#$Misclassified Simplices} & Error rate\\
 \cmidrule{5-8}
 & & & & Edges & Triangles & Tetrahedra & Total & \\
\toprule
\multirow{8}{*}{0.0} 
& 1GRM & 260 & 932 & 0 & 0 & 0 & 0 & 0.0000\\
& 1U71 & 1505 & 5696 & 0 & 0 & 0 & 0 & 0.0000\\
& 3N0H & 1509 & 5739 & 0 & 0 & 0 & 0 & 0.0000\\
& 4HHB & 4384 & 38796 & 0 & 0 & 0 & 0 & 0.0000\\
& 2J1N & 8142 & 29642 & 0 & 0 & 0 & 0 & 0.0000\\
& 1K4C & 16068 & 62851 & 15 & 33 & 16 & 64 & 0.0010\\
& 2OAU & 16647 & 123175 & 12 & 21 & 5 & 38 & 0.0003\\
& 1AON & 58674 & 262244 & 22 & 39 & 21 & 82 & 0.0003\\
\midrule
\multirow{8}{*}{1.0} 
& 1GRM & 260 & 1598 & 0 & 0 & 0 & 0 & 0.0000\\
& 1U71 & 1505 & 10828 & 0 & 0 & 0 & 0 & 0.0000\\
& 3N0H & 1509 & 10965 & 0 & 0 & 0 & 0 & 0.0000\\
& 4HHB & 4384 & 65987 & 0 & 0 & 0 & 0 & 0.0000\\
& 2J1N & 8142 & 58205 & 0 & 0 & 0 & 0 & 0.0000\\
& 1K4C & 16068 & 118467 & 20 & 34 & 14 & 68 & 0.0006\\
& 2OAU & 16647 & 199101 & 10 & 22 & 10 & 42 & 0.0002\\
& 1AON & 58674 & 495683 & 10 & 26 & 21 & 57 & 0.0001\\
\bottomrule
\end{tabular}
\label{tab:wrongSimplices}
\end{table}
The proposed algorithm requires computation of $\mathsf{OrthoSize}$ for each simplex, which in turn requires solving systems of linear equations. These computations require higher precision than is available on the GPU. So, the results may contain numerical errors. These numerical errors ultimately manifest as misclassification of a simplex as belonging to $K_\alpha$ or not. We performed extensive experimentation and observed that the alpha complex was correctly computed in several cases. In cases where the results were not correct, the number of false positives and negatives (extra or missing simplices) is extremely small as compared to the total number of simplices in the alpha complex. We observed a worst case error rate of $0.001$  in our experiments, see Table~\ref{tab:wrongSimplices}. This error rate is tolerable for several applications. If exact computation is required, we could use a tolerance threshold to tag some simplices as requiring further checks, which can in turn be performed on the CPU or GPU using a multi-precision library~\cite{joldes2016campary}.
The use of multi-precision libraries will adversely affect the computation time. Our implementation can be extended to use such a hybrid strategy. However, it will require additional experimentation to identify appropriate tolerance thresholds and performance optimization. We plan to implement this in future.

\section{Conclusions}
We proposed a novel parallel algorithm to compute the alpha complex for biomolecular data that does not require prior computation of the complete Delaunay triangulation. The useful characterization of simplices that belong to the alpha complex may be of independent interest. The algorithm was implemented using CUDA, which exploits the characteristics of the atom distribution in biomolecules to achieve speedups of up to a factor of 22 compared to the state-of-the-art parallel algorithm for computing the weighted Delaunay triangulation, and up to a factor of 50 speedup over the state-of-the-art implementation that is optimized for biomolecules. In future work, we plan to further improve the runtime efficiency of the parallel implementation and to resolve the numerical issues using real arithmetic. 

Applications of alpha complex outside the domain of biomolecular analysis often require the complete filtration of Delaunay complex. The algorithm as presented here is not best suited for such cases. However, the algorithm may be modified to utilize a previously computed alpha complex to efficiently compute the alpha complex for higher values of $\alpha$. We plan to investigate this extension in future work.

\bibliography{SoCGFullVersion}

\end{document}